\documentclass[11pt]{article}
\usepackage[margin=1in]{geometry}

\usepackage{fullpage}
\usepackage[table,xcdraw]{xcolor}
\usepackage{amsmath, amsthm, amssymb}
\usepackage{algpseudocode,algorithm,algorithmicx}
\usepackage{algorithmicx}
\usepackage{mathtools}
\usepackage[numbers]{natbib}
\usepackage{comment} 
\usepackage[most]{tcolorbox}
\usepackage{xfrac}
\usepackage{hyperref,cleveref}
\usepackage{multirow}
\usepackage{caption}
\usepackage{bm}
\usepackage{newfloat}
\usepackage{enumitem}
\usepackage{xfrac}
\usepackage{bbm}
\usepackage{soul}
\usepackage{makecell}
\usepackage{wrapfig}
\usepackage{hhline}

\usepackage{thmtools}
\usepackage{thm-restate}

\usepackage{graphicx} 
\usepackage{color}
\usepackage{tikz}
\usepackage{pgfmath}
\usepackage{listofitems}
\usetikzlibrary{decorations.pathreplacing}

\usepackage{array}
\usepackage{adjustbox}
\usepackage{multirow}
\usepackage{makecell}
\usepackage{booktabs}

\usepackage[table]{xcolor}

\usepackage{tikz}
\usepackage{pgfplots}
\pgfplotsset{compat=1.17}
\usepackage{pgfplotstable}
\newcolumntype{P}[1]{>{\centering\arraybackslash}p{#1}}
\newcolumntype{M}[1]{>{\centering\arraybackslash}m{#1}}
\newcommand*\samethanks[1][\value{footnote}]{\footnotemark[#1]}

\newcommand{\cls}{\mathbf{cl}}

\newcommand*{\todos}{}%

\ifdefined\notodos
\def\neel#1{}
\def\yusuf#1{}
\fi

\ifdefined\todos

\ifdefined\todos

\def\neel#1{\marginpar{$\leftarrow$\fbox{N}}\footnote{$\Rightarrow$~{\sf\textcolor{purple}{#1 --Neel}}}}
\def\yusuf#1{\marginpar{$\leftarrow$\fbox{D}}\footnote{$\Rightarrow$~{\sf\textcolor{blue}{#1 --Yusuf}}}}
\fi

\title{Limitations of Stochastic Selection\\ with Pairwise Independent Priors
}
 \author{
	Shaddin Dughmi \thanks{Supported by NSF Grant CCF-2009060.} \\
	Department of Computer Science\\
	University of Southern California\\
	{\tt shaddin@usc.edu}
	\and
	Yusuf Hakan Kalayci \samethanks\\ 
	Department of Computer Science\\
	University of Southern California\\
	{\tt kalayci@usc.edu}
	\and
	Neel Patel \samethanks\\ 
	Department of Computer Science\\
	University of Southern California\\
	{\tt neelbpat@usc.edu}
}
\date{ }

\newtheorem{theorem}{Theorem}[section]

\newtheorem{lemma}[theorem]{Lemma}

\newtheorem{claim}[theorem]{Claim}
\newtheorem{definition}[theorem]{Definition}

\newtheorem{question}[theorem]{Question}
\newtheorem{observation}[theorem]{Observation}

\newtheorem{cor}[theorem]{Corollary}

\newcommand{\set}[1]{\left\{ #1 \right\}}
\newcommand{\union}{\cup}
\newcommand{\intersect}{\cap}

\newcommand{\sm}{\setminus}

\newcommand{\ceil}[1]{\left\lceil {#1} \right\rceil}

\renewcommand{\tilde}{\widetilde}
\renewcommand{\bar}{\overline}

\newcommand{\Ft}{\mathbb F_2}

\newcommand{\Ber}{\operatorname{Ber}}
\newcommand{\MtS}{\operatorname{MatrixToSet}}
\newcommand{\indicator}{\mathbbm 1}
\newcommand{\Ehard}{\EE_{\text{hard}}}

\newcommand{\opt}{\textbf{OPT}}

\DeclareMathOperator{\ind}{ind}

\DeclareMathOperator{\spn}{\textbf{Span}}

\DeclareMathOperator{\OPT}{OPT}

\newcommand{\rank}{\textbf{Rank}}

\let\vec\mathbf

\def\max{\qopname\relax n{max}}

\def\argmax{\qopname\relax n{argmax}}

\def\Pr{\qopname\relax n{\mathbf{Pr}}}
\def\Ex{\qopname\relax n{\mathbb{E}}}

\newcommand{\RR}{\mathbb{R}}

\newcommand{\ZZ}{\mathbb{Z}}

\def\A{\mathcal{A}}

\def\D{\mathcal{D}}
\def\E{\mathbb{E}}
\def\EE{\mathcal{E}}
\def\F{\mathcal{F}}
\def\FF{\mathbb{F}}

\def\I{\mathcal{I}}

\def\M{\mathcal{M}}

\def\P{\mathcal{P}}

\def\R{\mathbb{R}}

\def\w{\vec w}

\def\eps{\epsilon}
\def\sse{\subseteq}

\renewcommand{\vec}{\mathbf}
\newcommand{\eat}[1]{}

\newenvironment{lp*}{\begin{equation*}  \begin{array}{lll}}{\end{array}\end{equation*}}

\newcommand{\pwset}{\Delta_{\text{pw}}(2^E)}

\algnewcommand{\IIf}[1]{\State\algorithmicif\ #1\ \algorithmicthen}
\algnewcommand{\EndIIf}{\unskip\ \algorithmicend\ \algorithmicif}

\newcounter{proc}

\newenvironment{tbox}{
\vspace{0.2cm}
\begin{tcolorbox}[width=\textwidth,
                  enhanced,
                  boxsep=2pt,
                  left=1pt,
                  right=1pt,
                  top=4pt,
                  boxrule=1pt,
                  arc=0pt,
                  colback=white,
                  colframe=black,
                  breakable]
}{
\end{tcolorbox}
}

\newcommand{\tboxhrule}[0]{\vspace{0.1cm} \hrule \vspace{0.2cm}}

\newenvironment{titledtbox}[1]{\begin{tbox}#1 \tboxhrule}{\end{tbox}}

\newenvironment{procedure}[1]{\refstepcounter{proc}\begin{titledtbox}{\textbf{Procedure \theproc.} #1}}{\end{titledtbox}}

\newcolumntype{?}{!{\vrule width 1.3pt}}

\begin{document}

\maketitle

\begin{abstract}
Motivated by the growing interest in correlation-robust stochastic optimization, we investigate stochastic selection problems beyond independence. Specifically, we consider the instructive case of pairwise-independent priors and matroid constraints. We obtain essentially-optimal bounds for contention resolution and prophet inequalities.  The impetus for our work comes from the recent work of \citet{pi-uniform-prophet}, who derived a constant-approximation for the single-choice prophet inequality with pairwise-independent priors. 

For general matroids, our results are tight and largely negative. For both contention resolution and prophet inequalities, our impossibility results hold for the full linear matroid over a finite field. We explicitly construct pairwise-independent distributions which rule out an $\omega\left(\frac{1}{\rank}\right)$-balanced offline CRS and an $\omega\left(\frac{1}{\log \rank}\right)$-competitive prophet inequality against the (usual) oblivious adversary. For both results, we employ a generic approach for constructing pairwise-independent random vectors --- one which unifies and generalizes existing pairwise-independence constructions from the literature on universal hash functions and pseudorandomness. Specifically, our approach is based on our observation that random linear maps turn linear independence into stochastic independence.

We then examine the class of matroids which satisfy the so-called partition property --- these include most common matroids encountered in optimization. We obtain positive results for both online contention resolution and prophet inequalities with pairwise-independent priors on such matroids, approximately matching the corresponding guarantees for fully independent priors. These algorithmic results hold against the almighty adversary for both problems.

\end{abstract}


\newpage 
\tableofcontents
\newpage

\section{Introduction}

Combinatorial optimization subject to uncertainty has gained substantial interest in recent years, initially motivated by its applications in computational economics (\cite{hartline2013bayesian,chawla2010multi}).
In many of these tasks, the underlying uncertainty or stochasticity arises from either the random availability of elements of a set system or from a stochastic weight assignment to these elements. Two fundamental {stochastic selection problems}, \emph{contention resolution} (e.g., \cite{vondrak2011submodular,adamczyk2018random,feldman2016online,pollner2022improved,nuti2022towards}) and \emph{generalized prophet inequalities} (e.g., \cite{kleinberg2019matroid,gravin2019prophet,dutting2020prophet}), fit into this paradigm. These problems appear either directly, or indirectly as subroutines, throughout the fields of algorithms and combinatorial optimization with a wide range of applications including approximation algorithms \cite{vondrak2011submodular,rubinstein2017combinatorial,gravin2019prophet}, mechanism design \cite{pollner2022improved,feldman2016online,ezra2020online,dutting2020prophet,alaei2014bayesian},  online algorithms \cite{srinivasan2023online,dughmi22,rubinstein2017combinatorial}, stochastic probing \cite{feldman2016online,adamczyk2015improved,baveja2018improved}, sparsification \cite{dughmi2023sparsification}, and algorithmic delegation \cite{kleinberg2018delegated,bechtel2020delegated,bechtel2022delegated}. 

A rich literature examines the design of algorithms for these problems when the input is a product distribution or negatively correlated. However, our understanding is relatively limited when the input distribution exhibits correlations, particularly positive correlations, which are often present in many intended applications. For instance, consider the scenario of sequential posted pricing where a seller with a single item encounters $n$ prospective buyers in sequence, each possessing a valuation for the product. The seller, with the goal of maximizing profit, offers a fixed, non-negotiable price to each buyer, who then decides to buy the item if the price is less than or equal to their valuation. Yet, in today's hyper-connected world, it is unrealistic to presume buyers remain unaffected by or ignorant of each other's valuations. In fact, notable studies \cite{park2007elaboration, askalidis-2016-value} demonstrate this phenomenon by showing that the aggregate online reviews from a large group of buyers play a critical role in shaping customer behavior.

A deeper understanding of the interplay between correlation and optimal selection, and an expansion of the algorithmic and complexity-theoretic toolkit thereof, promises to impact the myriad aforementioned applications of decision-making subject to uncertainty.  Of particular note is the matroid secretary conjecture of \citet{babaioff_secretary}, which has recently been shown equivalent to stochastic selection in the presence of a particular kind of positive correlation by Dughmi \cite{dughmi20,dughmi22}. Algorithmic approaches for near-optimal decision making in the presence of correlation, as well as proof techniques for ruling out such algorithms, could therefore shed light on the conjecture.

A number of recent works explore a variety of models in which decisions must be made in the presence of correlated inputs \cite{babaioff2020escaping-cannibalization, bei2019correlation-robust-single-item-auction, cai21simple-mechanisms-dependent-items, garvin2018correlation-robust-monopolist, bateni2015revenue, immorlica2020prophet}. It is either known (e.g. \cite{hill1992survey,rinott1992prophet-inequality-survey}), or easy to show, that not much can be achieved in the presence of arbitrary positive correlation. Even under assumptions like the \emph{linear correlation model} of \citet{bateni2015revenue}, in the worst case there are no positive algorithmic results for prophet inequalities with non-sparse dependencies even for the rank one matroid as shown by \citet{immorlica2020prophet}.

Particularly inspiring our investigation is the recent work by \citet{pi-uniform-prophet}, which initiates the study of stochastic selection problems with inputs that are \emph{pairwise independent}: any two random variables are independent, though positive or negative correlations can manifest when considering larger groups of variables.  Pairwise independence significantly relaxes the usual assumption of full independence, and  pairwise independent distributions have found application in hashing, derandomization, and constructions of pseudo-random generators (for more details, see surveys \cite{luby2006pairwise,salil2012pseudorandomness}). In the context of the sequential posted pricing mechanism, empirical studies \cite{park2007elaboration, askalidis-2016-value} emphasize the significant impact of a large number of aggregated online reviews on shaping customer behavior. These studies suggest that a buyer's valuation is influenced by reviews from a large number of consumers, whereas a small selection of customer reviews (which reflect their valuations) have little effect. As such, pairwise (or more generally $k$-wise) independence serves as a reasonable idealization of such settings where correlations live largely in the higher-order moments of a distribution.

\citet{pi-uniform-prophet} show that pairwise independence suffices for a constant approximation in the single-choice prophet inequality problem and sequential posted price mechanisms. This finding encourages further exploration of stochastic selection problems under the same pairwise independent assumption. Our focus is particularly on prophet inequalities and contention resolution schemes. This naturally leads us to the following question:

\begin{question}
    Do constant approximation prophet inequalities or contention resolution schemes exist for a broader class of set-systems when the input distribution is pairwise independent?
\end{question}

We resolve the above question for matroids. We prove strong impossibility results for matroid prophet inequalities and contention resolution schemes when the stochastic inputs are only pairwise independent. These impossibility results hold even for the most permissive computational models considered for these problems, and stand in contrast to the strong algorithmic results for inputs that are mutually independent \cite{kleinberg2019matroid,chekuri2011multi,lee2018optimal,feldman2016online}. The following summarizes our main contribution.

\begin{itemize}
    \item There is no $\omega\left(\frac{1}{\log \rank}\right)$-competitive matroid prophet inequality for pairwise independent distributions. This holds even for the \emph{oblivious adversary} who selects the order of elements in advance, and even for binary matroids.
    \item There is no $\omega\left(\frac{1}{\rank}\right)$-balanced contention resolution scheme for pairwise independent distributions. This holds even in the \emph{offline} setting of contention resolution, and even for linear matroids.
\end{itemize}
We complement these negative results with simple algorithms that  match these bounds when inputs are pairwise independent, even in the most restrictive of computational models considered for these problems: a $\Theta\left(\frac{1}{\log \rank}\right)$-competitive matroid prophet inequality  and a $\Theta\left(\frac{1}{\rank}\right)$-balanced \emph{online} contention resolution scheme. In contrast to our impossibility results, both our algorithmic results hold even for the \emph{almighty adversary} who selects the order of elements with knowledge of all the realized inputs and any internal randomness of the algorithm

For both of our impossibility results, we carefully construct a pairwise independent distribution for the linear matroid $\FF_q^d$ for some large $d \in \ZZ_{+}$ and a suitable prime $q$. Our approach to constructing pairwise independent distributions is founded on the observation that uniformly random linear maps between vector spaces convert linear independence in the domain space to stochastic independence in the range space. To put it formally, when a family of $k$-wise linearly independent vectors is embedded in another vector space via a uniformly random linear map, the embedded vectors exhibit $k$-wise stochastic independence and each assumes a uniform marginal distribution over the second vector space. Special instances of this observation have previously been employed to define $k$-wise independent hash functions \cite{CarterWe79, WegmanCa79} and $k$-wise independent random bits \cite{alon1986fast-maximal-independent-set, karloff1994construction-k-wise-independent, alon1990approx-k-wise-independent, naor1990approx-k-wise-independent}. For a comprehensive overview of prior work on the construction of pairwise independent distributions, we refer interested readers to the survey by \cite{luby2006pairwise} and to \cite[Chapter~3]{salil2012pseudorandomness}. Our method for constructing $k$-wise stochastically independent vectors can be viewed as a simple unification and vector-generalization of existing constructions of scalar-valued random variables.\footnote{Despite the simplicity of our construction, we have been unable to identify another construction with this level of generality. The concepts presented here permeate existing work on constructing $k$-wise independent random variables.}

Later, we examine the class of matroids that satisfy the partition property --- these include the most common matroids encountered in optimization. Informally, this property holds if a matroid can be approximated by a (random) partition matroid. We demonstrate that, when a matroid fulfills the partition property, we can reduce the problem to one defined over rank one matroids. Leveraging the results and machinery of \cite{pi-uniform-prophet}, we obtain constant factor prophet inequalities and contention resolution schemes for pairwise independent distributions on matroids satisfying the partition property. As in our previous algorithmic results, our bounds hold even for the almighty adversary for both problems.
We note the concurrent independent work of \citet{anupam-pairwise-indep-crs}, which also studies pairwise-independent stochastic selection. They obtain similar contention resolution schemes and prophet inequalities for a number of natural matroid classes such as the uniform, laminar, graphic, co-graphic, and regular matroids.

Finally, we mention that our results deepen the existing schism between matroids that admit the partition property and those that do not, and in doing so might shed light on the matroid secretary conjecture. Much of the interest in the partition property is due to the fact --- pointed out in a survey by \citet{dinitz2013recent} --- that matroids satisfying the $\alpha$ partition property also admit an $O(\alpha)$-competitive secretary algorithm. In fact, most classes of matroids for which constant-competitive secretary algorithms are known --- such as graphic \cite{babaioff2009secretary}, co-graphic \cite{soto2013matroid}, and laminar \cite{jaillet2013advances} ---  satisfy a constant partition property. Moreover, many such algorithms either explicitly or implicitly exploit the partition property.   \citet{dinitz2013recent} therefore asked whether  every matroid satisfies a constant partition property, as a possible route to resolving the matroid secretary conjecture. This question was answered negatively by the recent work of \citet{shayan-partition-matroid-secretary}, who show that the full binary matroid of rank $d$ does not satisfy the $\alpha$-partition property for any $\alpha \leq O(d^{1/4})$. The parallel work of \citet{singla-matroid-secretary-barriers} also provides evidence of the limitations of partition-based algorithms for the secretary problem.

Our results add to this literature in two distinct ways. First, as corollaries of our results we show that full linear matroid $\FF_q^d$ with $q \geq d$ does not satisfy the $\alpha$ partition property for any $\alpha \leq O(d)$, and  that the full binary matroid of rank $d$ does not satisfy the partition property for any $\alpha \leq O(d/\log d)$,  strengthening the bound of \cite{shayan-partition-matroid-secretary}.  Second, by constructing provably ``hard instances'' for selection problems that are easy in the presence of the partition property, our techniques might shed light on the analogous question for the matroid secretary problem. In fact, showing that our construction for prophet inequalities remains hard in the random order model would disprove the matroid secretary conjecture. On the flip side, providing an algorithm for our construction in the random order model appears highly nontrivial, and therefore might stimulate the development of algorithmic techniques pertinent to the conjecture.

\subsection*{How to Read this Paper}
We present pertinent preliminaries in Section~\ref{sec:prelims}, of which a light perusal is  sufficient for the reader comfortable the basics of matroid theory, contention resolution, and prophet inequalities.  We then present an abridged technical overview of our results and techniques in Section~\ref{sec:overview}, followed by a more detailed treatment in Sections~\ref{sec:tool} through \ref{sec:partition-property and Pw Selection}. The reader looking to get a high-level sense of our results and techniques is invited to focus primarily on Section~\ref{sec:overview},  referring to the later technical sections for more detail as needed.  We close with open questions in Section~\ref{sec:conclusion}.


\newpage

\section{Preliminaries}
\label{sec:prelims}
\subsection{Basic Notation and Terminology} \label{sec:notations}
We use bold lowercase letters to denote vectors,  with the $\ell^{\text{th}}$ component of a vector $\vec v$ denoted by $\vec v(\ell)$.  Sets and matrices are denoted by uppercase letters, while collections of sets or matrices are denoted by bold uppercase letters. For a set $S$ of indices into a vector $\vec v$, we denote $\vec v(S) = \sum_{\ell \in S} \vec v (\ell)$. We denote the set of positive integers up to $n$ by $[n]$. For a set $X$, we denote its power set by $2^X$, and use $\Delta(X)$ to denote the family of all distributions supported on $X$.

A \emph{set system} is a pair $\M=(E,\I)$ where $E$ is a set of \emph{elements} and $\I \sse 2^E$ is a family of \emph{feasible} (a.k.a. \emph{independent}) sets. We concern ourselves primarily with set systems where $E$ is finite, and $\I$ is \emph{downwards-closed}: If $B \in \I$ and $A \sse B$ then $A \in \I$.  For a vector $\vec \mu \in [0,1]^{E}$ of \emph{marginals} indexed by the elements, we use $\Delta_{pw}(2^{E})(\vec \mu) \subseteq \Delta(2^E)$ to denote the family of pairwise independent distributions over sets of elements with marginal probabilities $\vec \mu$, i.e. for every $\mathcal D(\vec \mu)\in \Delta_{\text{pw}}(2^E)(\vec \mu)$, $\Pr_{Q\sim \D(\vec \mu)}[e\in Q] = \mu(e)$ and the events $\{e\in Q\}_{e\in E}$ are pairwise independent. 
We also consider distributions in $\Delta(\mathbb R_{\geq 0}^{E})$ that assign a nonnegative weight (a.k.a. value) to each element. We let $\Delta_{\text{pw}}(\R_{\geq 0}^{E}) \subseteq \Delta(\R_{\geq 0}^{E})$ be the class of pairwise independent weight distribution over elements $E$ --- i.e., if  $\vec w \sim \D \in \Delta_{\text{pw}}(\RR_{\geq 0}^E)$ then for any distinct pair of elements $e,f \in E$ their weights  $\vec w (e)$ and $\vec w (f)$ are pairwise independent random variables. Throughout the paper, we interchangeably use the terms ``weight'' and ``value''.

For a prime number $q$, $\FF_q$ denotes the finite field with $q$ elements and $\FF_q^d$ denotes the vector space of dimension $d$ over $\FF_q$.\footnote{Unless otherwise specified, we think of  vectors as column vectors.} For a given set of integer labels $L \subseteq \ZZ$, we denote the collection of labeled vectors by $\FF_q^d \times L := \set{\vec v^i: \text{for all }i\in L \text{ and } \vec v \in \FF_q^d }$.
We use capital letters to symbolize matrices over these finite fields, and their rank is denoted by $\rank(\cdot)$. A matrix $R \in \FF_q^{r \times c}$ is a full column-rank matrix if $\rank(R)=c$, i.e. its columns constitute a set of linearly independent vectors. Additionally, we often refer to the columns of a matrix $R \in \FF_q^{d \times n}$ using lowercase letters and subscripts, such as $\vec r_1, \dots \vec r_n \in \FF_q^d$, and we occasionally allow ourselves some flexibility in notation (clarifying with re-declarations in context) to use the matrix $R$ interchangeably with the set of its columns. For any matrix $A$, we denote its column space as $\cls(A)$.

\subsection{Matroid Theory}\label{sec:matroid-prelims}
We use standard definitions from matroid theory; for details see \cite{oxley2006matroid,welsh2010matroid}. A \emph{matroid} $\M=(E,\I)$ is a set-system with elements $E$ and a family of independent sets $\I \subseteq 2^E$ satisfying the three \emph{matroid axioms}. A \emph{weighted matroid} incorporates a matroid $\M=(E, \I)$ with weights $w \in \R^E$ for its elements.

By duplicating or making parallel labeled copies of each element of a matroid $\M=(E, \I)$ ``$m$'' times, we construct a larger matroid $\M^{\times m}=(E^{\times m}, \I^{\times m})$. Here $E^{\times m}$ contains $m$ parallel copies  $e^1,\ldots,e^m$ of  each $e \in E$, and $T \subseteq E^{\times m}$ is in $\I^{\times m}$ if $\{e : e^i \in T \mbox{ for some i} \} \in \I$ and $|T\cap \{e^i: i\in [m]\}|\leq 1$ for all $e\in E$.

The \emph{rank function} of matroid $\M=(E, \I)$ is denoted by $\rank^\M$, where $\rank^\M(S) = \max\{|T| : T\subseteq S, T \in \I\}$. The \emph{weighted rank function}  $\rank_{\vec w}^\M$ is defined for weighted matroids $(\M,\vec w)$ as $\rank^\M_{\vec w}(S) = \max\{\vec w(T): T\subseteq S, T \in \I\}$. The span function of matroid $\M$ is denoted by $\spn^\M(S)$ where $\spn^\M(S)=\{e \in E: \rank^\M(S \cup \{e\})=\rank^\M(S)\}$. We slightly abuse notation and use $\rank(\M)=\rank^\M(E)$ for the rank of the matroid. We may omit the superscript $\M$ when it is clear from  context. 

The \emph{matroid polytope} $\P(\M) \sse [0,1]^E$ associated with $\M$ is the convex hull of all indicator vectors of its independent sets. Equivalently, a nonnegative vector $\mu$ is in $\P(\M)$ if and only if $\sum_{e \in S} \mu_e \leq \rank^\M(S)$ for all $S \subseteq E$. 

A \emph{linear matroid} $\M=(E,\I)$ is a matroid where $E$ is a family of vectors in some vector space, and $\I$ consists of the linearly-independent subsets of $E$. We consider linear matroids where the underlying vector space is $\FF_q^d$, for $q$ a prime and $d \in \mathbb N$.  When $q=2$, this is also referred to as a \emph{binary matroid}. When $E=\FF_q^d$ we call this the \emph{full linear matroid over $\FF_q^d$}, and  when $E=\FF_2^d$ we call it the \emph{full binary matroid with rank $d$}.

The \emph{rank one matroid} on elements $E$ is the matroid whose independent sets are the singletons in $E$ as well as the empty set. A \emph{simple partition matroid} on $E$ is the disjoint union of rank one matroids; i.e., there is a partition $E_1,\ldots,E_k$ of $E$ such that $S$ is independent if and only if $|S \intersect E_i| \leq 1$ for all $i$.

\subsection{Contention Resolution}

\emph{Contention Resolution} is the algorithmic task of converting a random set which is feasible  ``on average'' to on to one which is always feasible. Contention resolution in the offline setting was originally formalized by \citet{vondrak2011submodular} for application to approximation algorithms. It has since been generalized, studied, and applied in various online models (see  \cite{feldman2016online,adamczyk2018random,lee2018optimal}).
An algorithm for contention resolution is referred to as a \emph{Contention Resolution Scheme (CRS)}.

Fix a downward-closed set-system $\M=(E,\I)$ over a ground set $E$ of elements, as well as a \emph{convex relaxation} $\P(\M) \sse [0,1]^E$ of the indicator vectors of sets in $\I$. Let  $\vec \mu \in \P(\M)$ be a vector of marginal probabilities, and let $A \sse E$ be a set of \emph{active} elements drawn from a known distribution $\D \in \Delta(2^E)$ satisfying $\Pr_{A \sim \D}[e\in A] = \mu(e)$ for all $e\in E$. Given these inputs, the goal of a CRS is to \emph{select} (a.k.a. \emph{accept}) a feasible subset of the active elements --- i.e., a set $I \in \I$ such that $I \subseteq A$.  

A CRS is judged by its balance ratio: we say that a contention resolution scheme is $c$-balanced if for all $e\in E$, $\Pr[e \in I \mid e\in A] \geq c$. Many natural classes of combinatorial constraints, including matroids, matchings, and knapsacks,  admit $\Omega(1)$-balanced contention resolution schemes when the events $\{e \in A\}_{e\in E}$ are jointly independent, and $\P(\M)$ is the usual relaxation of the problem.

In the \emph{offline} model of contention resolution, all inputs --- in particular the set $A$ of active elements --- are given upfront. This is the most permissive model considered for contention resolution, and serves as the setting of our impossibility results.  Our algorithmic results hold for the more restrictive  \emph{online} setting where elements are presented to the algorithm, which we refer to as an \emph{online CRS (OCRS)}, in some order determined by an adversary. When $e \in E$ arrives online, it is then revealed whether $e$ is active (i.e., whether $e \in A$), at which point the algorithm must irrevocably decide whether to select $e$ subject to feasibility. Several adversary models have been considered for online contention resolution, and our algorithmic results hold even for the most restrictive of those: the \emph{almighty adversary} who determines the order of elements with full knowledge of all inputs including the realization of $A$, as well as the realization of any internal randomness used by the algorithm. 

In this work, we focus on contention resolution for matroids when $\{e \in A\}_{e\in E}$ are only pairwise independent, and $\P(\M)$ is simply the matroid polytope. In particular, $\M=(E,\I)$ is a matroid and  $\D \in \Delta_{\text{pw}}(2^E)(\vec x)$ for some $\vec x \in \P(\M)$. When a CRS is $c$-balanced for all such $\D$, we say it is a \emph{$c$-balanced pairwise-independent CRS}.

Finally, we restate a Theorem from \cite{dughmi20} that characterizes the set of distributions which permit balanced contention resolution schemes in the offline setting.

\begin{theorem}[Theorem~3.6 from \cite{dughmi20}]\label{thm:shaddins characterization}
	Fix a matroid $\M=(E, \I)$, and let $\D$ be a distribution supported on $2^E$. The following are equivalent for every $c \in [0,1]$,
	\begin{enumerate}
		\item There exists an offline contention resolution scheme which is $c$-balanced for $\D$. 
		\item For every weight vector $\w \in \R_{\geq 0}^E$, the following holds: $ \Ex_{A \sim \D}[\rank_{\w}(A)] \geq c \cdot \Ex_{A \sim \D}[\w(A)]$ 
        \item For every $F \subseteq E$, the following holds: $ \Ex_{A \sim \D}[\rank(A \intersect F)] \geq c \cdot \Ex_{A \sim \D}[|A \intersect F|]$.
	\end{enumerate}
\end{theorem}

\subsection{Prophet Inequalities}
Fix a downwards-closed set system $\M=(E, \I)$. In a  \emph{prophet inequality problem}, there are \emph{weights} (or \emph{values}) $\w \in \R_{\geq 0}^E$ on the elements that are drawn from a distribution $\D$, and elements arrive online in some order determined by an adversary. 
We take the perspective of a \emph{gambler} who a-priori knows $\mathcal M$ and the distribution $\D$ of weights, but not the realized weights $\w$. When an element $e$ arrives online, the gambler learns $\vec w(e)$ and must irrevocably decide whether or not to accept $e$, subject to accepting a feasible set of of elements $S \in \I$. The gambler seeks to maximize their utility $\w(S) = \sum_{e \in S} \w(e)$. The goal is to compete --- in expectation --- with an omniscient \emph{prophet} who obtains the maximum possible utility $\max\{ \w(T) : T \in \I\}$. When the gambler's expected utility is an $\alpha$ fraction of the prophet's expected utility, we say that we have an $\alpha$-competitive prophet inequality for $\M$ and $\D$. 

Our negative results in this paper hold even against the weakest of adversaries considered in the literature on  prophet inequality problems: the \emph{oblivious adversary} who determines the order of elements in advance as a function of only $\M$ and $\D$; the gambler, therefore, knows the (arbitrary) order at the outset. In contrast, our positive results hold even for the strongest adversary considered in the literature: the \emph{almighty adversary} who determines the order of elements with full knowledge of all inputs, including the realized weights $\w$ as well as any internal randomness of the gambler's algorithm. The gambler therefore only learns the order as elements arrive online.

In this paper, we focus on  prophet inequality problems where $\M$ is a matroid and $\D$ is a pairwise independent distribution over weight vectors , i.e. $\D \in \Delta_{\text{pw}}\left(\R_{\geq 0}^{|E|}\right)$.

\subsection{A Useful Lemma  for Pairwise Independent Events}

The following lemma, from \cite{pi-uniform-prophet} establishes a lower bound for the probability that at least one event from a collection of pairwise independent events will occur.
\begin{lemma}[Lemma 1 from \cite{pi-uniform-prophet}]\label{lem:local-lemma-type}
	Let $\D$ and $\ind$ be pairwise independent and mutually independent distributions over a collection of random events $\{\mathcal E_i\}_{i=1}^k$ such that $\Pr_{\EE \sim \D}[\EE_i] = \Pr_{\EE \sim \ind}[\EE_i]$. Then,
	\begin{equation*}
		\Pr_{\vec \EE \sim \D}\left[\bigvee_{i=1}^k \mathcal E_i\right] \geq \frac{\sum_{i=1}^k \Pr[\mathcal E_i]}{1+\sum_{i=1}^k \Pr[\mathcal E_i]} \quad \text{  and  } \quad\Pr_{\EE \sim \D}\left[\bigvee_{i=1}^n E_i\right] \geq \frac{1}{1.299} \cdot \Pr_{\EE \sim \ind}\left[\bigvee_{i=1}^n \EE_i \right]
	\end{equation*}
	where $\EE \vee \F$ denotes the event that at least one of $\EE$ or $\F$ occurs.
\end{lemma}


\section{Overview of Technical Results}
\label{sec:overview}
In this section, we present an overview of our techniques and results. First we outline our construction of pairwise stochastically independent vector families. We then show how to employ this construction to prove our main impossibility results for pairwise-independent offline matroid contention resolution, and pairwise-independent matroid prophet inequalities against the oblivious adversary. We show that our impossibility results are tight by providing algorithms with matching bounds for both problems, even against the almighty adversary. Finally, we examine matroids satisfying the partition property, which includes most common matroids encountered in combinatorial optimization. We provide constant factor algorithms for pairwise-independent contention resolution and prophet inequalities on such matroids, even against the almighty adversary.

\subsection{A Recipe for Pairwise-Independent Vector Families}
\label{sec:overview-recipe}

As our main technical tool, we present a simple and flexible recipe for constructing pairwise independent families of vectors. We instantiate this recipe in different ways for matroid prophet inequalities and contention resolution schemes. Since our construction permits encoding of a rich variety of higher order relationships between the vectors, while maintaining lower-order independence, we hope it might be of independent interest.

There are two versions of our recipe. The first version produces an \emph{ordered} family (i.e., a tuple) of pairwise independent vectors, and the second turns that into an \emph{unordered} family (i.e., a set) of \emph{labeled} vectors wherein membership is pairwise independent. We require the second, unordered, version for our impossibility results. The ordered construction is the most natural, and easily generalizes to $k$-wise independence for arbitrary $k$ --- we present the more general construction in this paper. The unordered construction is built on its ordered counterpart, and is tailored for pairwise independence (i.e., $k=2$).

\subsubsection*{Ordered Families}
We begin with our construction of an \emph{ordered} family of $k$-wise independent vectors over the field $\FF_q$, where $q$ is a prime. Let $m,n,d \geq k$ be positive integers. Let $\Sigma \in \FF_q^{m \times n}$ be a matrix whose columns $\vec \sigma_1, \ldots, \vec \sigma_n \in \FF_q^m$ are $k$-wise \emph{linearly} independent; i.e., no linear combination of $k$ or fewer of these columns evaluates to the zero vector in $\FF_q^m$. Intuitively, $\Sigma$ is the input matrix we get to ``design'' for encoding a ``hard instance'' of the problem at hand, while respecting $k$-wise independence.  Let $R \in \FF_q^{d \times m}$ be drawn uniformly at random; i.e., each entry of the matrix $R$ is a uniformly-random element of $\FF_q$.  Let $X = R \Sigma \in \FF_q^{d \times n}$ be the output matrix, with columns $\vec x_1,\ldots,\vec x_n \in \FF_q^d$ where $\vec x_i = R \vec \sigma_i$.

The key observation here, which we prove in Section~\ref{sec:tool}, is that the columns $\vec x_1,\ldots,\vec x_n \in \FF_q^d$ of $X$ are $k$-wise \emph{stochastically} independent, and moreover each is uniformly distributed in $\FF_q^d$. Notably, the uniformly-random linear map $R$ converts linear independence to stochastic independence. Also notably, linearity of $R$ entails that any higher-order (greater than $k$) linear relationships between the columns of $\Sigma$ --- designed to inject ``hardness'' as previously described--- are preserved as relationships between the corresponding columns of $X$. In both our applications, the dimension $d$ of the output vectors $\vec x_1,\ldots, \vec x_n$ is larger than the dimension $m$ of the inputs $\vec \sigma_1,\ldots, \vec \sigma_n$, implying that the linear operator $R$ is injective --- and the $\vec x_i$s are therefore distinct --- with high probability.

\subsubsection*{Unordered Families}
Both of our impossibility results require constructing a set of non-zero vectors $A$ where the events $\set{\vec v \in A}$ are pairwise independent, and yet feature higher-order positive dependencies. This is easiest to see in the case of contention resolution, where we require the set of active elements to be concentrated in a lower-dimensional subspace, as needed to rule out a balanced CRS. This motivates our second, \emph{unordered}, construction, which we describe next.

A natural first attempt would be to invoke our ordered construction to obtain pairwise-independent  $[\vec x_1, \ldots, \vec x_n ] =  R  \Sigma$, then take $A= \set{\vec x_1,\ldots,\vec x_n} \sm \vec 0$ to be the (unordered) set of non-zero columns. Some thought reveals that this can introduce mild pairwise correlation. 
This is largely because $\vec u \in A$ implies one fewer of the $n$ ``chances'' is available for a different vector $\vec v$. Moreover, analyzing the exact magnitude of this correlation is complicated by the event --- albeit a low probability one in our applications --- that the random linear operator $R$ is non-injective. We circumvent these issues by creating $n$ copies $\vec v^1,\ldots,\vec v^n$ of each vector $\vec v \in \FF_q^d$, and including $\vec v^i$ in $A$ when $\vec x_i = \vec v \neq 0$. Therefore, we work in the matroid $\FF_q^d \times [n]$, with $n$ parallel copies of each element in the matroid $\FF_q^d$ \emph{labeled} with the positions $1,\ldots,n$.  
It is now clear that each non-zero $\vec u^i$ is in $A$ with probability $1/q^d$, and that the events $\vec u^i \in A$ and $\vec v^j \in A$ are pairwise-independent so long as $i \neq j$. This leaves the case of $\vec u^i$ and $\vec v^i$ for distinct $\vec u$ and $\vec v$, whose membership in $A$ is mutually exclusive and hence negatively correlated. This, however, can be easily corrected by mixing in --- with small probability $\frac{1}{q^d}$ ---  a set which positively correlates vectors with the same label without introducing dependencies across different labels, nor changing the marginals.\footnote{It is this step which recovers pairwise, but not necessarily $k$-wise, independence. A generalization to $k$-wise independence appears more technically involved, though likely possible.} We describe the details in Section~\ref{sec:tool}.\footnote{Our procedure might seem similar to that of \citet{alon2003almost}, where they devise a procedure to convert an \emph{almost} $k$-wise independent distribution to a $k$-wise independent distribution over $\{0,1\}^n$. However,  given pairwise independent random variables (or events) $X_1,\dots, X_n$, their procedure requires the following: for any $S\subseteq I$ with $|S|\leq k$, $\Pr[\oplus_{i\in S} X_i  = 1] = \frac{1 +\eps}{2}$ for small $\eps >0$ which is clearly not the case in our setting as $\Pr[\vec v^i \in A \land \vec u^i \in A] = 0$ for any distinct $\vec u, \vec v\in \FF_q^d$.  Here $\oplus$ denotes the binary sum or XOR of the bits.}

\subsection{Contention Resolution on Matroids}
\label{sec:technical-crs}

We use our construction of pairwise independent vector families to rule out a balance ratio better than $\frac{3}{\rank}$ for pairwise-independent contention resolution on linear matroids, even in the offline setting.  For a desired rank $d$, we instantiate the recipe described in Section~\ref{sec:overview-recipe} with $m=2$, $n=d$, and an arbitrary prime $q \geq d$. The input matrix $\Sigma \in \FF_q^{2 \times d}$ is a fat (rank $2$) matrix with $d$ pairwise-linearly-independent columns $\sigma_1,\ldots,\sigma_d \in \FF_q^2$ --- we show that such a matrix exists whenever $q \geq d$. The random linear operator $R \in \FF_q^{d \times 2}$ then maps $\Sigma$ to $X \in \FF_q^{d \times d}$. The $d$ columns $\vec x_1, \ldots, \vec x_d \in \FF_q^d$ of $X$ are uniformly distributed in $\FF_q^d$, stochastically pairwise independent, and most crucially --- since $\Sigma$ has rank $2$ --- all lie in a subspace of rank at most $2$.

We convert $X=[\vec x_1,\ldots, \vec x_d]$ to a pairwise-independent set $A \sse \FF_q^d \times [d]$ as described in Section~\ref{sec:overview-recipe}. Each non-loop element of the matroid $\FF_q^d \times [d]$ is in $A$ with probability $\frac{1}{q^d}$, so its marginals lie in the matroid polytope\footnote{To see this, note that a uniformly random base of $\FF_q^d \times [d]$ has essentially the same --- in fact, very slightly larger --- marginals since all the vectors are equally likely to be part of the base and $\vec 0$ is never sampled.} and $\Ex[ |A|]  \approx d$. Moreover, since $X$ has rank at most $2$, and $A$ consists of (labeled copies of) the columns of $X$ with probability $1-\frac{1}{q^d}$, we can bound $\Ex[\rank(A)] \leq 2+ \frac{d}{q^d} < 3$. It follows from Theorem~\ref{thm:shaddins characterization} that no contention resolution scheme has a balance ratio better than $\frac{\Ex[\rank(A)]}{\Ex[|A|]} \leq \frac{3}{d}$, as claimed. This holds even in the offline setting of contention resolution.

We show that our bound of $O(1/\rank)$ is essentially tight for general matroids. Given a pairwise-independent distribution with feasible marginals for a matroid of rank $d$, the contention resolution scheme which greedily selects active elements with probability $\frac{1}{2d}$ is $\frac{1}{4d}$-balanced, even against the almighty adversary. This follows simply from pairwise independence, as well as the fact that the marginals sum to at most $d$.  

We also note that our CRS impossibility result for linear matroids can easily be adapted to the important special case of binary matroids,  while degrading the bound from $O\left(\frac{1}{\rank}\right)$ to $O\left(\frac{\log \rank}{\rank}\right)$. Specifically, we let $\Sigma \in \FF_2^{O(\log d) \times d}$ be a fat (rank $O(\log d)$) binary matrix with $d$ pairwise independent columns, which we show always exists. We then proceed in identical fashion with a uniformly random linear map $R \in \FF_2^{d \times \log d}$, culminating  in a stochastically pairwise-independent family of $\approx d$ vectors in $\FF_2^d$  with expected rank $O(\log d)$, as needed.

\subsection{Prophet Inequalities on Matroids}\label{sec:technical-pw}

Through a somewhat more involved application of our recipe, we rule out a competitive ratio better than $O(\frac{1}{\log d})$ for prophet inequalities on binary matroids of rank $d$, even against the oblivious adversary. We also show by way of an algorithm that this bound is tight for pairwise-independent prophet inequalities, even for general matroids and against the almighty adversary. 

\subsubsection*{The Impossibility Result}

We begin with the following randomized construction. For arbitrarily large $d$ and some $\kappa = \Omega(\log d)$, we  define a (random) nested sequence  $V_1 \supset  V_2 \ldots \supset V_{\kappa}$ of subspaces of $\FF_2^d$,  and corresponding independent sets $S_1, \ldots, S_\kappa \sse \FF_2^d$,  satisfying the following properties:

\begin{enumerate}[label=(\roman{enumi})]
\item $V_\ell$ has dimension $\frac{d}{2^{\ell -1}}$\label{prop:flat_dim}
\item $S_\ell$ is a linearly independent subset of $V_\ell$ with size $n_\ell$ equal to a constant fraction of its dimension.\footnote{In fact, we can guarantee half the dimension, i.e. $n_\ell = d/2^\ell$.}\label{prop:basis}
\item Let $\ell< \ell' \leq \kappa$. Conditioned on $S_1,\ldots,S_\ell$ and $V_1,\ldots,V_\ell$, a vector $v \in S_\ell$ is in $V_{\ell'}$ with probability $\geq 1/2^{\ell'-\ell}$. \label{prop:span}
\item The sets $S_1, \ldots, S_\kappa$ are disjoint.\label{prop:disjoint}
\end{enumerate}

We show the existence of such random $S_\ell$s and $V_\ell$s through a highly technical construction, the details of which we defer to Section~\ref{sec:PI_main}. For now, the reader might be satisfied of its plausibility by noting that it becomes trivial if we drop property \eqref{prop:disjoint}: Let $S_1$ be the standard basis vectors, then for each $\ell$ let $S_{\ell+1}$ be a random half of $S_\ell$, and let $V_\ell$ be the span of $S_\ell$. We cannot help but speculate whether our construction, which goes to some pains in order to achieve all four properties simultaneously, can be simplified or elegantly reduced to known linear-algebraic facts.

To motivate this construction, consider the following ``hard'' instance of the prophet inequality problem on the matroid $\FF_2^d$ where elements of each $S_\ell$ are assigned weight $w_\ell = 2^\ell$, and remaining elements are assigned weight $0$. Suppose also that the non-zero weight elements arrive in increasing order of weight, i.e. in the order $S_1,\ldots,S_\kappa$, followed by the zero-weight elements at the end. Forgive for a moment that the weights are not pairwise independent, and that guaranteeing this order is beyond the powers of the oblivious adversary --- these are issues we will address later. The maximum weight independent set, which can be constructed by running the greedy algorithm on the non-zero weight elements in decreasing order of weights $S_\kappa, \dots, S_1$, selects a constant fraction of each $S_\ell$ -- this follows easily from the matroid exchange axiom and Property \ref{prop:basis}. Therefore, the prophet's reward is $\Omega(\sum_{\ell=1}^\kappa w_\ell \cdot n_\ell) = \Omega(d \log d)$. The gambler, in having to choose a subset $T_\ell$ of each $S_\ell$  before learning anything about flats $V_{\ell'}$ for $\ell' > \ell$, is not as fortunate. Properties \eqref{prop:flat_dim} and \eqref{prop:span} imply that  $\Ex[\sum_{\ell=1}^{\ell'} |T_\ell| \cdot \frac{1}{2^{\ell'-\ell}} ] \leq \frac{d}{2^{\ell'-1}}$. Plugging in $w_\ell = 2^\ell$ and multiplying both sides by $2^{\ell'}$,  we get $\Ex[\sum_{\ell=1}^{\ell'} w_\ell |T_\ell|] \leq 2 d$. Since this holds for arbitrary $\ell'$, we conclude that the gambler's expected total reward is only $O(d)$.

It remains to address two issues with this prophet inequality instance: pairwise independence, and the limited powers of the oblivious adversary. For pairwise independence, we use our recipe described in Section~\ref{sec:overview-recipe}. We let $\Sigma_\ell$ be the binary $d \times n_\ell$ matrix with $S_\ell$ as its columns, and let the matrix $\Sigma = [ \Sigma_1,\ldots,\Sigma_\kappa] \in \FF_2^{d \times n}$ with $n=\sum_{\ell=1}^\kappa n_\ell$ be the concatenation of the $\Sigma_\ell$s. By property~\ref{prop:disjoint}, as well as the fact that we defined each $S_\ell$ as a set (rather than a multiset), the columns of $\Sigma$ are distinct. Since we are in $\FF_2$, this is equivalent to the columns being pairwise linearly independent, as required for using $\Sigma$ in our recipe. For the random linear embedding $R$, we choose a slightly larger output dimension to guarantee that the mapping is injective with high probability, and therefore preserves the geometry of $S_1,\ldots,S_\kappa$ as captured by properties \eqref{prop:flat_dim} through \eqref{prop:disjoint}. An output dimension of $2d$ suffices, so we let $R$ be a uniformly random matrix in $\FF_2^{2d \times d}$.

We proceed as described in Section~\ref{sec:overview-recipe}. We define $X = R \Sigma \in \FF_2^{2d \times n}$, or in more detail $X= [X_1,\ldots,X_\kappa]$ where $X_\ell = R \Sigma_\ell \in \FF_2^{2d \times n_\ell}$ has as its columns the image of $S_\ell$ under the linear map $R$. Recall from Section~\ref{sec:overview-recipe} that, for a fixed $\Sigma$, any pair of columns of $X$ are distributed independently and uniformly in the destination space $\FF_2^{2d}$. It follow that this continues to hold when using our randomly-constructed $\Sigma$, and the columns of our matrix $X$ constitute a pairwise-independent ordered family of vectors distributed uniformly in $\FF_2^{2d}$. We then convert $X$ to an unordered family  $A \sse \FF_2^{2d} \times [n]$ wherein membership is stochastically pairwise independent as in Section~\ref{sec:overview-recipe}. In more detail, with probability $1-1/2^d$  we let $A_\ell$ consist of the the $n_\ell$ columns of $X_\ell= R \Sigma_\ell$ 
labeled with the integers $L_\ell$ from $(\sum_{\ell' <\ell} n_{\ell'}) +1$ to $\sum_{\ell' \leq \ell} n_{\ell'}$,
and let $A$ be the (disjoint) union of the $A_\ell$s. With remaining probability $1/2^d$ the set $A$ is drawn from a positively-correlated distribution designed to ensure pairwise-independence overall as described in Section~\ref{sec:overview-recipe}.

In summary, $A$ is an (unordered) family of (labeled) vectors in $\FF_2^{2d} \times [n]$ wherein membership is stochastically pairwise independent. Moreover, with high probability $A$ is the disjoint union of $A_1,\ldots,A_\kappa$, where $A_\ell$ is simply $S_\ell$ transformed by an injective linear map $R$ then distinctly labeled with integers from $L_\ell$. By assigning weight $w_\ell = 2^\ell$ to elements with labels in $L_\ell$, we obtain what is effectively our original ``hard'' instance of the matroid prophet inequality transformed by the injective linear operator~$R$. Moreover, since a labeled vector $\vec v^i \in \FF_2^{2d} \times [n]$ with $i \in L_\ell$ has weight $w_\ell=2^\ell$ when $\vec v^i \in A$ and weight $0$ otherwise, it follows that the weights are stochastically pairwise  independent. Therefore, we have converted our original hard instance into one which is stochastically pairwise independent.

But what of the oblivious adversary's power to set the arrival order? Fortunately, now that the weight of $\vec v^i \in \FF_2^{2d} \times [n]$  is either zero or uniquely determined as a non-decreasing function of its its label $i \in [n]$, ordering elements in increasing order of label guarantees that vectors in $\FF_2^{2d}$ arrive in non-decreasing order of weight. In particular, after arrival of elements with labels in $L_1,\ldots,L_\ell$, the gambler knows the positions in $\FF_2^{2d}$ of the vectors with weights up to $2^\ell$ --- those corresponding with high probability to the original  $S_1,\ldots,S_\ell$ --- but nothing else of the positions of vectors with weights exceeding $2^\ell$  --- corresponding to $S_{\ell'}$ for $\ell' > \ell$.  This is a fixed arrival order, and therefore obviously within the power of the oblivious adversary.

\subsubsection*{The Algorithm}

We show that our bound of $O(1/\log \rank)$ is essentially tight for general matroids. Specifically, we obtain an $\Omega(1/\log \rank)$ matroid prophet inequality for pairwise independent distributions, even against the almighty adversary. We necessarily exploit pairwise independence, as no nontrivial guarantee is possible for general correlated distributions even for the rank one matroid.

Our algorithm is based on simple bucketing. Let $\M=(E,\I)$ be a matroid with pairwise-independent random weights $w \in \RR_{\geq 0}^\E$, and let $\OPT= \Ex[\rank_w(\M)]$ denote the expected reward of the prophet. Weights smaller than $\frac{OPT}{2 \rank}$ contribute at most half the prophet's reward, so can be discarded. Partition the weights between $\frac{OPT}{2 \rank}$ and $3 \OPT$ into $O(\log \rank)$  ``regular'' buckets delimited by the integer powers of $2$.  Weights larger than $3 OPT$ are assigned their own ``special'' bucket. The algorithm simply chooses the bucket with the largest contribution to the prophet's reward, and greedily selects as many elements as possible from that bucket subject to feasibility.

Clearly, the greedy algorithm applied to any of the regular buckets recovers at least half of that bucket's contribution to the prophet's expected reward. The special bucket $B_\infty$ requires a more careful analysis which exploits pairwise independence of the events $e \in B_\infty$.  By Markov's inequality, $B_\infty$  is non-empty with probability at most  $1/3$.   A Lemma from \cite{pi-uniform-prophet} for pairwise-independent events implies that $\sum_{e \in E} \Pr[ e \in B_\infty]  \leq \frac{\Pr [ B_\infty \neq \emptyset]}{1-\Pr[B_\infty \neq \emptyset]} \leq \frac{1}{2}$. Invoking pairwise independence again, together with  the union bound, an element $e$ in $B_\infty$ is the only element in that bucket with probability at least $1/2$, even after conditioning on its weight $w_e$. It follows that the greedy algorithm applied to $B_\infty$ recovers half of its contribution to the prophet's reward.

Putting it all together, since our algorithm greedily selects from the bucket contributing most to the prophet's reward, and there are $O(\log \rank)$ buckets, we obtain an $\Omega(1/\log \rank)$ prophet inequality for pairwise-independent distributions on matroids. Since we assumed nothing about the order in which elements are greedily selected, this holds even against the almighty adversary. 

\subsection{Exploiting the Partition Property}

In contrast to our negative results for general matroids, we show that not all is lost for most common matroids in the optimization literature. Such matroids often satisfy the  \emph{$\alpha$-partition property} for a constant $\alpha$. Informally, this means that the matroid $\M$ can be approximated by a randomly-chosen simple partition matroid $\M'$ on the same elements, in the sense that the weighted rank of $\M'$ approximates that of $\M$ from below up to a factor of $\alpha$ for every vector of element weights. 

For prophet inequalities, we apply the pairwise-independent single-choice prophet inequality of \cite{pi-uniform-prophet} to each part of $\M'$ separately. This gives a $\frac{1}{3}$  prophet inequality for $\M'$, and therefore a $\frac{\alpha}{3}$ prophet inequality for $\M$, for pairwise independent distributions. Since the guarantees of \cite{pi-uniform-prophet} hold even against the almighty adversary, so do ours.

For contention resolution, the situation is slightly more involved. First, we show the existence of a balanced offline CRS for pairwise-independent distributions by utilizing Theorem \ref{thm:shaddins characterization} and Lemma \ref{lem:local-lemma-type}, then we show how to exploit a duality argument as well as the aforementioned prophet inequality to convert it to an online CRS against the almighty adversary.

Let $A$ be a set of elements wherein membership is pairwise-independent with marginals $\mu$ feasible for $\M$. To show the existence of an offline balanced CRS, it suffices by Theorem \ref{thm:shaddins characterization} to show, for every set of elements $F$, that the expected rank of $A \cap F$ is a constant fraction of its expected size $\mu(F)$. We begin by analyzing the set $\tilde{A}$ where membership is jointly independent with the same marginals $\mu$, and later relate $A$ and $\tilde{A}$ through Lemma \ref{lem:local-lemma-type}. It follows from \cite{vondrak2011submodular} that  $\tilde{A}$ admits a $(1-1/e)$-balanced CRS with respect to $\M$, and therefore $\Ex[\rank_\M(\tilde{A} \intersect F)] \geq (1-1/e) \cdot \mu(F)$ for every set of elements $F$. Invoking the partition property, we get $\Ex[\rank_{\M'}(\tilde{A} \intersect F)] \geq \alpha \cdot (1-1/e) \cdot \mu(F)$. We then observe that the rank function of the partition matroid $\M'$ decomposes additively across its parts $E_1,\ldots,E_k$, with the $i$-th part contributing $1$ to $\rank_{\M'} (\tilde{A} \intersect F)$  precisely when at least one of the elements in $E_i \cap F$ is in $\tilde{A}$. Lemma \ref{lem:local-lemma-type} implies that the probability of a disjunction of pairwise-independent events approximates, up to a factor of $1.299$, the same quantity in the jointly independent case. In particular, $\Ex[\rank_{\M'}(A \intersect F)] \geq \frac{1}{1.299} \cdot \alpha \cdot (1-1/e) \cdot \mu(F)$. Since the rank function of $\M'$ is smaller than that of $\M$, we can invoke Theorem \ref{thm:shaddins characterization} to conclude the existence of an $\frac{1-1/e}{1.299} \cdot \alpha$-balanced CRS for $A$ with respect to $\M$. 

For good measure, we also show how to turn our offline CRS into an online one, even against the almighty adversary, at a cost of an additional $O(\alpha)$ factor in the balance ratio. We use a duality-based construction essentially identical to that in \cite[Theorem 4.1]{dughmi20}. At a high level, the problem of maximizing the weighted rank of active elements functions as a dual to contention resolution. An online $\beta$-approximation to this dual problem can be converted to an online CRS whose balance ratio is within $\beta$ of the best possible offline. Our prophet inequality is such an approximation with $\beta = \frac{\alpha}{3}$, even against the almighty adversary. Therefore, we obtain a pairwise-independent online CRS with balance ratio $\frac{1-1/e}{3.897} \cdot \alpha^2$  against the almighty adversary.

We note that our offline CRS for matroids satisfying the partition property, together with the impossibility results for linear and binary matroids in Section~\ref{sec:technical-crs}, imply limits on the partition property for these matroids. In particular, for an $\alpha$-partition property  we show that $\alpha = O\left(\frac{1}{d}\right)$  for the full linear matroid, and  $\alpha = O\left(\frac{\log d}{d}\right)$ for the full binary matroid, where $d$ denotes the rank. The latter results improves on the bound of $O\left(\frac{1}{d^{1/4}}\right)$ from \cite{shayan-partition-matroid-secretary}.


\section{A Recipe for Pairwise-independent Vector Families}
\label{sec:tool}

In this section, we devise our main tool to produce a family of pairwise independent (labeled) vectors.
Initially, we detail our approach to constructing an ordered family of pairwise independent vectors. Subsequently, we outline the process for transforming an ordered family to an unordered family.
	
\subsection{Ordered Pairwise Independent Vector Families}

    The construction of an ordered family of pairwise independent vectors naturally extends to the $k$-wise independence. Therefore, in this section, we present the construction in the most general form.
    We fix positive integers $m,n,d \geq k$. Let $\Sigma \in \FF_q^{m \times n}$ be a matrix with $k$-wise linearly independent columns $\vec \sigma_1 , \vec \sigma_2, \dots \vec \sigma_n$ and $R \in \FF_q^{d\times m}$ be a matrix where each entry is a uniformly random element of $\FF_q$ drawn independently. Consider the matrix $X = R \Sigma \in \FF_q^{d\times n}$ whose columns are obtained by linearly transforming each column of $\Sigma$ via $R$. In the following lemma, we show that the columns $\vec x_1 ,\dots , \vec x_n \in \FF_q^{d}$ of $X$ form an ordered family of  $k$-wise stochastically independent vectors., i.e., for each $S \subseteq [n]$ of size at most $k$ the vectors $\{\vec x_i \}_{i\in S}$ are mutually (stochastically) independent.

\begin{lemma}
	\label{lem:rotation}
	Let $\Sigma \in \FF_q^{m \times n}$ be a matrix with $k$-wise linearly independent columns $\sigma_1, \dots \sigma_n \in \FF_q^m$, $R\in \FF_q^{d \times m}$ be a uniformly random matrix with entries $r_{i,j} \sim \operatorname{Unif}\left( 0,1,\dots ,q-1\right)$, and $X =  R \Sigma$ with column vectors $\vec x_1, \dots \vec x_n$. For any subset $S \subseteq[n]$ of size at most $k$, the vectors $\set{\vec x_i}_{i\in S}$ are mutually stochastically independent. Moreover, for any $i\in [n]$ and $\vec v\in \FF_q^d$, $\Pr[\vec x_i = \mathbf v] = \frac{1}{q^d}$.
\end{lemma}

\begin{proof}
First we show that each column $\vec x_i \in \FF_q^d$ of $X$ is uniformly distributed over $\FF_q^d$. For arbitrary $\vec v \in \FF_q^d$, we can express,
\begin{align*}
	\Pr[\vec x_i = \vec v] &= \Pr[R\vec \sigma_i = \vec v] \notag \\
	 & = \Pr \left[\bigwedge_{\ell=1}^d \left\{ \sum_{j=1}^m r_{\ell, j} \cdot \vec  \sigma_i(j) = \vec v(\ell) \right\} \right] \notag \\
	&=\prod_{\ell=1}^d \Pr \left[ \sum_{j=1}^m r_{\ell, j} \cdot \sigma_i(j) = \vec v(\ell)\right]&& (\text{disjoint set of independent RVs}). 
\end{align*}
For any fixed $\ell \in [d]$, the set of possible solutions $(r_{\ell, 1},\dots , r_{\ell,m})$ to the equation $\sum_{j=1}^m r_{\ell, j} \cdot \sigma_i(j) = \vec v_{i}(\ell)$, considering $\vec \sigma_i$ and $\vec v$ as fixed, forms an affine subspace of rank $m-1$ which has size of $q^{m-1}$. As each $r_{i,j}$ is a uniformly random element from $\FF_q$, the probability,
\begin{equation}
\Pr[\vec x_i = \vec v]  = 	\prod_{\ell=1}^d \Pr \left[ \sum_{j=1}^m r_{\ell, j} \cdot \sigma_i(j) = \vec v(\ell)\right] = \left (\frac{q^{m-1}}{q^m} \right)^d= \frac 1 {q^d}. \label{eq:unifomrly_random_vector} 
\end{equation}

Next, we prove $k$-wise stochastic independence of the columns of $X$. Let $S \sse [n]$ be such that $|S| \leq k$. By assumption, the corresponding columns $\set{\sigma_i}_{i \in S}$ of $\Sigma$ are linearly independent. It follows that

\begin{align*}
  \Pr\left[ \bigwedge_{i \in
  S} \{\vec x_i = \vec v_i\} \right] \notag
	&= \Pr \left[ \bigwedge_{i \in S} \left\{ R \vec \sigma_i = \vec v_i \right\} \right] \notag \\
	&= \Pr \left[ \bigwedge_{i \in S} \bigwedge_{\ell=1}^d \left\{ \sum_{j=1}^m r_{\ell, j} \cdot \vec  \sigma_i(j) = \vec v_{i}(\ell) \right\} \right] \notag \\
	&= \Pr \left[ \bigwedge_{\ell=1}^d \bigwedge_{i \in S} \left\{ \sum_{j=1}^m r_{\ell, j} \cdot \sigma_i(j) = \vec v_{i}(\ell) \right\} \right] \notag  \\
	&= \prod_{\ell=1}^d \Pr \left[ \bigwedge_{i \in S} \left\{ \sum_{j=1}^m r_{\ell, j} \cdot \sigma_i(j) = \vec v_{i}(\ell) \right\} \right]. && (\text{disjoint set of independent RVs}) 
\end{align*}

Consider the following system of equations with variables $(r_{\ell,1}, \dots r_{\ell,m})$ for a fixed $\ell \in [d]$:
\begin{equation}\label{eq:system_of_eqns}
    \sum_{j=1}^m r_{\ell, j} \cdot \sigma_i(j) = \vec v_{i}(\ell) \qquad \forall i \in S.
\end{equation}
Since  $\{\sigma_i \mid i \in S\}$ are linearly independent, the set of solutions $(r_{\ell,1}, \dots r_{\ell,m})$ to the system of equations~\eqref{eq:system_of_eqns} forms an affine subspace of dimension $m-|S|$ containing $q^{m-|S|}$ many vectors. As each $r_{i,j}$ is a uniformly random element from $\FF_q$, the probability,

\begin{align*}
\Pr\left[ \bigwedge_{i \in S} \{ \vec x_i = \vec v_i\} \right] &= \prod_{\ell=1}^d \Pr \left[ \bigwedge_{i \in S} \left\{ \sum_{j=1}^m r_{\ell, j} \cdot \sigma_{i}(j) = \vec v_{i}(\ell) \right\} \right]\\
& = \left(\frac{q^{m-|S|}}{q^m}\right)^d \\
&= q^{-d\cdot |S|}\\
& =  \prod_{i \in S} \Pr[ \vec x_i = \vec v_i] .
\end{align*}
Above, the second equality follows from the fact that $r_{i,j}$ is a uniformly random element from $\FF_q$ and the last equality follows from Equation~\eqref{eq:unifomrly_random_vector}. This proves that the random vectors $\{\vec x_i\}_{i\in S}$ are mutually independent, as needed.
\end{proof}

Next, we present a well-known fact about random matrices defined over finite fields. For the self containment of the paper, we provide its simple proof in Appendix~\ref{sec:missingproofsTool}. 

\begin{restatable}{lemma}{lemLinearIndep}
\label{lem:linear_indep}
Let $R\in \FF_q^{d \times m}$ be a uniformly random matrix with entries $r_{i,j} \sim \operatorname{Unif}\left( 0,1,\dots ,q-1\right)$. Then, for any $m<d$, we have $\Pr[\rank(R)=m] \geq 1-\frac{1}{q^{d-m}}.$ In addition, for any $\Sigma \in \FF_q^{m\times n}$ with distinct columns, the columns of $X = R\Sigma$ are distinct with probability $\geq 1-\frac{1}{q^{d-m}}$.
\end{restatable}

The last lemma essentially states that when $d>>m$, the random matrix $R \in \FF_q^{d \times m}$ becomes full rank. Therefore, the \emph{random linear map} under $R$ becomes a random linear ``embedding'' (or an injective mapping)  with high probability. In such a scenario, when a matrix $\Sigma \in \FF_q^{m \times n}$ consists of distinct columns, the product $X = R \Sigma$ will have distinct columns with high probability.

\subsection{Unordered Pairwise Independent Vector Families}
In the context of stochastic selection, we typically focus on probability distributions over sets of elements. This requires us to construct an (unordered) set of vectors $A$ such that for any two distinct vectors $\vec u$ and $\vec v$, the events $\set{\vec u \in A}$ and $\set{\vec v\in A}$ are independent.

In the previous section, we introduced a technique for producing a sequence $\vec{x}_1, \ldots, \vec{x}_n \in\FF_q^d$ of vectors that exhibit pairwise (indeed $k$-wise) stochastic independence.
To turn this into a set of vectors $A$, it is tempting to define $A = \set{\vec x_1,\ldots,\vec x_n}$. However, some thought reveals that membership in this set $A$ can exhibit some slight negative correlation. Even in the (high probability) event that the vectors $\vec x_1, \ldots, \vec x_n$ are distinct, for distinct $\vec u, \vec v \in \FF_q^d$ we have  $ \Pr[\vec u \in A] = \Pr[ \vec v \in A]  = \frac{n}{q^d}$, and yet $\Pr [ \vec u \in A \land \vec v \in A ] =  \frac{n (n-1)}{q^{2d}} < \Pr [\vec u  \in A ] \Pr [ \vec v \in A] $. It might be tempting to ``correct'' for this negative correlation in $A$ by slightly mixing with a different set $A'$  featuring  positive correlation that is calibrated to ``cancel out'' the negative correlation in $A$. However, this is further complicated by the (low probability) event where the random linear map used in our construction is singular, and the vectors $\vec x_1, \ldots, \vec x_n$ are therefore not necessarily distinct. This makes it challenging to quantify the exact ``amount'' of correlation in $A$ in general, and therefore challenging to argue for the existence of a suitable $A'$. To circumvent these difficulties, we take a somewhat different approach which creates $n$ labeled duplicates for each vector. We describe this approach next. 

Recall the construction of ordered families from the previous section:  The matrix $\Sigma \in \mathbb{F}_q^{m \times n}$ is comprised of pairwise linearly independent column vectors $\sigma_1, \dots, \sigma_n$,  the matrix  $R \in \FF_q^{d \times m}$ has its entries $r_{i,j}$ drawn i.i.d. uniformly from $\FF_q$, and the matrix  $X=[\vec x_1, \ldots, \vec x_n] = R \cdot \Sigma$ is the image of $\Sigma$ under the linear transformation defined by matrix $R$. We create $n$ labeled copies $\vec v^{\ell_1} , \dots , \vec v^{\ell_n}$ of each vector $\vec v \in \FF_q^d$,  and include $\vec v^{\ell_i} \in A$ whenever the $i$-th column of $X$ equals $\vec v$, i.e. $\vec x_i =\vec v$. This allows the inclusion of each vector $\vec v$ in the set $A$ multiple times, each occurrence distinguished by a unique \emph{label} from a set of labels $L$. $A$ is therefore a subset of $\FF_q^d \times L$. Without loss of generality, one can think of $L = [n]$.

\begin{figure}
\begin{procedure}{Random Set Generation from Random Matrix (MatrixToSet)}\label{proc:construct_pw_set}
    \textbf{Input:} $X \in \FF_q^{d \times n}$ with pairwise independent columns $\vec x_1, \dots \vec x_n$ and an ordered label set $L = \{\ell_1 , \dots , \ell_n\}$,
    
    \textbf{Output:} $A \subseteq \FF_q^{d} \times L$.

    \textbf{Dist - I : } $\mathcal D_1$
     \begin{enumerate}[label={(\arabic*)}]
		\item Construct an unordered set $A: = \{\vec x_i^{\ell_i} : i\in [n] \} \subseteq \FF_q^d \times L$. 
     \end{enumerate}

    \textbf{Dist II : }$\mathcal D_2$
    \begin{enumerate}[label={(\arabic*)}]
        \item Start with $A=\emptyset$. Then, for each label $\ell_i \in [n]$, independently with probability $\frac{1}{q^d}$, include all elements $\{ \vec v^{\ell_i} : \vec v\in \FF_q^d \}$ in $A$.
    \end{enumerate}
   
\textbf{Sample} $A \subseteq \FF_q^d \times L$ from $\mathcal D_1$ w.p $1-\frac{1}{q^{d}}$ and $\mathcal D_2$ w.p. $\frac 1 {q^d}$.
\end{procedure}
\end{figure}
Let us proceed to understand correlation structure within this set $A$. First, we observe that for distinct indices $i,j \in [n]$ and any two (possibly equal) vectors $\vec v, \vec u \in \FF_q^d$, the events $\set{\vec v^{\ell_i} \in A}$ and $\set{\vec u^{\ell_j} \in A}$ are independent as events $\set{\vec x_i=\vec v}$ and $\set{\vec x_j=\vec u}$ are assured to be independent by Lemma~\ref{lem:rotation}. Nonetheless, this method introduces correlation between the events $\{\vec v^{\ell_i} \in A\}$ and $\{\vec u^{\ell_i} \in A\}$ since $\Pr[\vec u^{\ell_i} \in A \wedge \vec v^{\ell_i} \in A] = 0$ for any two distinct vectors $\vec v, \vec u \in \FF_q^d$ and index $i \in [n]$ (or label $\ell_i \in L$).
To overcome this, we mix $A$ with another distribution that positively correlates the inclusion of labeled vectors with the same label without introducing dependencies across different labels, nor changing the marginals. We describe our procedure of converting ordered families of vectors to unordered families of vectors in Procedure~\ref{proc:construct_pw_set}. 

In the subsequent lemma, we demonstrate that the unordered family of vectors sampled according to Procedure~\ref{proc:construct_pw_set} is pairwise independent provided the ordered family of vectors $\vec x_1,\dots , \vec x_n$ are pairwise independent and uniformly sampled from $\FF_q^d$. To maintain a smooth and uninterrupted discussion, we have postponed the detailed technical proof to Appendix~\ref{sec:missingproofsTool}.

\begin{restatable}{lemma}{lemSetConstruct}\label{lem:setConstruct}
    Let $X \in \FF_q^{d \times n}$ be a random matrix with pairwise independent column vectors $\vec x_1, \dots \vec x_n$ for some $n < q^d$ where each vector $\vec x_i$ is distributed uniformly in $\FF_q^d$. Then a random set $A \subseteq \FF_q^d \times L$ generated by Procedure~\ref{proc:construct_pw_set} given the inputs of $X$ and $L=\set{\ell_1, \dots \ell_n}$ satisfies:
    \begin{enumerate}[label={(\arabic*)}]
    	\item For any $\vec v^{\ell_i} \in \FF_q^d \times L$, $\Pr[\vec v^{\ell_i}  \in A]= \frac{1}{q^d}$.
    	\item For any two distinct $\vec v^{\ell_i}, \vec u^{\ell_j} \in \FF_q^d \times L$, the events $\{\vec v^{\ell_i} \in A\}$ and $\{ \vec u^{\ell_j} \in A\}$ are independent. 
     \end{enumerate}
\end{restatable}


\section{Pairwise-independent Contention Resolution on Matroids}\label{sec:pw-CRS-hardness} 
    In this section, we utilize the tool devised in the previous section to show the limits of pairwise-independent contention resolution schemes for matroids. In particular, we show that the class of linear matroids does not admit an $\omega(1/\rank)$-balanced pairwise independent offline CRS. We complement this  with a pairwise independent online CRS against the almighty adversary whose balance ratio matches our impossibility result for offline schemes up to a constant factor. 

\subsection{Limits of Pairwise Independent Contention Resolution}

We consider the matroid $\M^{\times d} = (\FF_q^d \times [d], \I^{\times d})$ which consists of $d$ labeled copies of each element in the full linear matroid $\FF_q^d$.
The following theorem states the limits of pairwise independent offline contention resolution for this matroid.

\begin{theorem}\label{thm:CRSmain}
    For any $d>2$ and any prime $q$, the matroid $\FF_q^d \times [d]$ does not admit a $\frac{c+1}{d}$-balanced pairwise independent offline CRS where $c$ is any positive integer satisfying $q^{c-1} \geq d$. 
    In particular,
    \begin{enumerate}
        \item $\FF_q^d \times [d]$ does not admit a $\frac{3}{d}$-balanced pairwise independent offline CRS for $q \geq d$.
        \item $\FF_2^d \times [d]$ does not admit a $\frac{3+\log_2 d}{d}$-balanced pairwise independent offline CRS. 
    \end{enumerate}
\end{theorem}

Our construction of a hard offline contention resolution instance utilizes the tool devised in Section~\ref{sec:tool} to sample pairwise independent unordered family of vectors. This unordered family will serve as the set of active elements for contention resolution. We first form a ``fat'' matrix $\Sigma \in \FF_q^{c \times d}$ ($c<<d$) with pairwise linearly independent columns. The existence of such a matrix can be assured by choosing $q$ and $c$ so that $d \leq q^{c-1}$ as we show in Claim~\ref{claim:sigma_crs_pw_linear_indep}.  
Then, linear transformation of $\Sigma$ via a uniformly random matrix $R \in \FF_q^{d \times c}$ yields a matrix $X=R \Sigma$ whose columns $\vec x_1, \dots \vec x_d$ are pairwise stochastically independent  due to Lemma~\ref{lem:rotation}.
Finally, we use Procedure~\ref{proc:construct_pw_set} to turn the sequence $(\vec x_1, \dots \vec x_d)$ into the set  $A \subseteq \FF_q^d \times [d]$ of active elements. Lemma~\ref{lem:setConstruct} ensures that for any two distinct elements $\vec v^i , \vec u^i \in \FF_q^d \times [d]$, the events $\{\vec v^i \in A\}$ and $\{\vec u^j \in A\}$ are pairwise independent. We summarize this construction in Procedure~\ref{proc:construct_crs_active_elements}.
\begin{figure}
\begin{procedure}{Pairwise Independent Set of Active Elements in $\M^{\times d} = (\FF^{d}_q \times [d], \I^{\times d})$}\label{proc:construct_crs_active_elements}
    \textbf{Input: } Positive integers $d,c$ and a prime $q$ satisfying $q^{c-1} \geq d$.

    \textbf{Output: } $A \subseteq \FF_q^d \times [d]$.

    \begin{enumerate}[label={(\arabic*)}]
        \item Let $\Sigma^{\text{CRS}} \in \FF_q^{c \times d}$ be an arbitrary matrix with pairwise linearly independent column vectors.
        \item Let $R\in \FF_q^{d \times c}$ be a matrix where each entry $r_{ij}$ is uniformly and independently sampled from $\set{0, \dots q-1}$.
        \item Set $A \leftarrow \MtS(R\cdot \Sigma^{\text{CRS}} , [d])$.
    \end{enumerate}
\end{procedure}
\end{figure}

The following claim together with Lemma~\ref{lem:setConstruct} implies that the output of Procedure~\ref{proc:construct_crs_active_elements} follows a pairwise independent distribution over $\FF_q^{d}\times [d]$.
\begin{claim}
    \label{claim:sigma_crs_pw_linear_indep}
    Given that $d \leq q^{c-1}$, there exists a matrix $\Sigma^{\text{CRS}} \in \FF_q^{c \times d}$ which consists of $d$ pairwise linearly independent columns.
\end{claim}
\begin{proof}
    Notice that there are $q^{c}-1$ non-zero vectors in $\FF_q^c$, and linear dependence partitions them into equivalence classes of size $q-1$. Therefore, by selecting one vector from each equivalence class one can generate $\frac{q^{c}-1}{q-1}>q^{c-1}$ pairwise linearly independent columns. When $d \leq q^{c-1}$, the matrix $\Sigma^{\text{CRS}}$ always exists.
\end{proof}

Before we show that $A$, the output of Procedure~\ref{proc:construct_crs_active_elements}, does not admit a $\frac{c+1}{d}$-balanced contention resolution scheme, we first verify that the marginals of $A$ reside within the matroid polytope $\P(\M^{\times d})$.
\begin{lemma}\label{lem:A in polytope}
    Let $\vec \mu \in \R_+^{\FF_q^d \times [d]}$ be the marginal probability vector of $A$ where $\mu(\vec v^i)=\Pr[\vec v^i \in A]$ when $A$ is sampled according to Procedure~\ref{proc:construct_crs_active_elements}. Then,
    $\vec \mu \in \P(\M^{\times d})$. 
\end{lemma}
\begin{proof}
    By Lemma~\ref{lem:setConstruct} we know that $\mu(\vec v^i) = \frac{1}{q^{d}}$ for each $\vec v^i \in \FF_q^d$. Thus, for any subset $S \subseteq \mathbb F_q^{d} \times [d]$, observe that 
    $$\mu(S):= \sum_{\vec v^i \in S} \mu(\vec v^i) = |S| \cdot \frac{1}{q^{d}} \leq \frac{d \cdot q^{\rank(S)}}{q^{d}} = \frac{d \cdot q^{\rank(S)}}{\rank(S) \cdot q^{d}} \cdot \rank(S)\leq \rank(S).$$
    Above, the first inequality holds because $|S| \leq d\cdot q^{\rank(S)}$ as we have $d$ copies of each $\vec v\in \FF_q^{d}$ in $\mathbb F_q^{d} \times [d]$. The last inequality follow from the fact that $x q^{-x}$ is a decreasing function of $x$ and $\rank(S) \leq d$, hence, $\frac{d \cdot q^{\rank(S)}}{\rank \cdot q^{d}} \leq 1$. Thus, $\vec \mu \in \P(\M^{\times d})$.
\end{proof}

Now we are ready to prove Theorem~\ref{thm:CRSmain}.
\begin{proof}[Proof of Theorem~\ref{thm:CRSmain}]
    Let $q$ be any prime, $c$ and $d$ be two integers satisfying $q^{c-1} \geq d$, and $A$ be the random set sampled according to Procedure~\ref{proc:construct_crs_active_elements}. 
    By Lemma~\ref{lem:setConstruct} and Claim~\ref{claim:sigma_crs_pw_linear_indep} we know that $A$ follows a pairwise independent distribution and Lemma~\ref{lem:A in polytope} demonstrates that the set $A$ has marginals inside the matroid polytope $\P(\M^{\times d})$. 
    
        The final step of Procedure~\ref{proc:construct_crs_active_elements} invokes Procedure~\ref{proc:construct_pw_set}. Recall that Procedure~\ref{proc:construct_pw_set} mixes distributions $\D_1$ and $\D_2$. Then,     
    \begin{align*}
        \Ex[\rank(A)] &\leq \Ex[\rank(A) \mid A \sim \D_1] \cdot \Pr[A \sim \D_1] + \rank(\M^{\times d}) \cdot \Pr[A \sim \D_2]\\
        &\leq \Ex[\rank(A) \mid A \sim \D_1] + \frac{d}{q^d} && (\rank(\M^{\times d})=d)\\
        &\leq c + 1 && (q \geq 2).
    \end{align*}
    Moreover, due to Lemma~\ref{lem:setConstruct} we know that $\Pr[\vec v^i \in A] = \frac{1}{q^d}$ for any $\vec v^i \in \FF_q^d \times [d]$. Thus, 
    $$ \Ex[|A|] = \sum_{\vec v^i \in \FF_q^{d} \times [d]} \Pr[\vec v^i \in A] = \frac{1}{q^{d}} \cdot \left|\FF_q^{d} \times [d] \right| = d.$$
    Combining these two facts we obtain:
    $$ \frac{\Ex[\rank(A)]}{\Ex[|A|]} \leq \frac{c+1}{d}.$$
        Finally, the characterization of distributions permitting contention resolution  from Theorem~\ref{thm:shaddins characterization} concludes the proof.
    
\end{proof}

\subsection{Optimal Pairwise Independent OCRS}

Consider a matroid $\M=(E, \I)$ and a randomly chosen set of active elements $A \subseteq 2^E$, which is sampled from a pairwise independent distribution $\D \in \Delta_{\text{pw}}(2^E)(\vec \mu)$ with marginals $\vec \mu \in \P(\M)$. In online contention resolution, the elements $e \in E$ arrive in some order chosen by an adversary, at which point the CRS algorithm learns whether $e$ is in $A$ and must irrevocably decide whether or not to select  $e$ for inclusion in its solution.  The algorithm is required to adhere to the feasibility constraints of the matroid during these inclusions.
We assume that the arrival order is chosen by the almighty adversary who is aware of the algorithm and all random outcomes, including those internal to the algorithm.

In the previous section we showed that for a class of linear matroids no offline CRS can attain a balance ratio greater than $\frac{3}{\text{rank}(\mathcal{M})}$ for pairwise independent distributions. The following theorem shows the existence  a $\frac{4}{\text{rank}(\mathcal{M})}$-balanced online CRS against the almighty adversary for all matroids.

\begin{theorem}
    For any matroid $\M=(E, \I)$, an arbitrary vector $\vec \mu \in \P_{\mathcal I}$ and distribution $\mathcal D \in \Delta_{\text{pw}}(2^E)(\vec \mu)$, there exists $\frac{1}{4\rank(\M)}$-balanced pairwise independent online contention resolution scheme against the almighty adversary.
\end{theorem}
\begin{proof}
    We use the simple greedy algorithm that selects each active element with a probability of $\frac{1}{2 \cdot \text{rank}(\mathcal{M})}$ unless it violates feasibility. This algorithm internally flips an independent random coin for each element $e \in E$, where each coin lands on ``heads'' with  probability $\frac{1}{2 \cdot \text{rank}(\mathcal{M})}$. We denote the joint occurrence of an element being active and its corresponding coin flip resulting in heads by $\EE_e$, and the complement of this event is denoted by $\bar{\EE_e}$. 

    We note two crucial observations: First, events $\set{\EE_e}_{e \in E}$ are pairwise independent since events $\set{e \in A}_{e \in E}$ are pairwise independent. Second, given that the marginal probability vector $\vec{\mu}(e) := \Pr[e \in A]$ is within the matroid polytope $\mathcal{P}(\mathcal{M})$, the total sum of probabilities $\sum_{e \in E} \Pr[e \in A]$ is at most  $\text{rank}(\mathcal{M})$. Consequently, the sum $\sum_{e \in E} \Pr[\EE_e]$ is at most $\frac{1}{2}$. We now bound the balance ratio of the greedy algorithm as follows.
    
\begin{align*}
    \Pr[e {\mbox{ is selected}}] &\geq \Pr\left[\EE_e \wedge \bigwedge_{f \in E \setminus \set{e}} {\bar{\EE_f}}\right] && (\EE_e\text{ solely occurs})\\
    &={\left(1 - \Pr\left[\bigvee_{f \in E \setminus \set{e}} \EE_f \middle| \EE_e\right] \right)} \cdot \Pr[\EE_e]\\
    &\geq \left(1-\sum_{f \in E\setminus \set{e}} \Pr[\EE_f \mid \EE_e]\right) \cdot \Pr[\EE_e]&& \text{(Union bound)}\\
    &= \left(1-\sum_{f \in E\setminus \set{e}} \Pr[\EE_f]\right) \cdot \Pr[\EE_e] && (\text{Pairwise Independence})\\
    &\geq \frac{1}{2} \cdot \frac{1}{2 \rank} \cdot \Pr[e \in A] && \left( \sum_{f \in E} \Pr[\EE_f] \leq \frac{1}{2}\right)\\
    &= \frac{1}{4 \rank} \cdot \Pr[e \in A]. && \qedhere 
\end{align*}
\end{proof}


\section{Pairwise-independent Prophet Inequalities on Matroids}
\label{sec:PI_main}

In this section, we explore a specific instance of the pairwise independent prophet inequality problem which demonstrates the impossibility of achieving any $\omega\left(\frac{1}{\log \rank}\right)$-competitive algorithm against the oblivious adversary. Our focus is on a complete binary matroid of rank $2d$, comprising $n$ labeled copies of each element, represented as $\M^{\times n} = (E^{\times n}, \I^{\times n})$. Here, $d$ is an integer power of $2$, $n= \Theta(d)$,  and $E^{\times n}$ is defined as $\FF_2^{2d}\times [n]$. The precise value of $n$ will be determined in subsequent discussions. The hard instance is structured around two main elements: (i) a pairwise independent probability distribution for the weights assigned to $E^{\times n}$ and (ii) a fixed order $\lambda$ of elements $E^{\times n}$. This setup will ensure that any algorithm adhering to $\lambda$ for processing elements inevitably fails to select an independent set that offers substantial reward, particularly due to the lack of prior knowledge regarding the weight assignments.

\subsection{Construction of Weight Distribution and Arrival Order}

We start by outlining the process for {constructing} a pairwise independent weight distribution for the elements in the matroid $\mathbb{F}_2^{2d} \times [n]$, where $n=\Theta(d)$.
Reflecting on the core principles summarized in Section~\ref{sec:technical-pw}, our approach for determining weight distribution involves generating $\kappa = \Theta(\log d)$ distinct \emph{levels}. 
The elements of $\FF_2^{2d} \times [n]$ will be partitioned into these levels based on their respective labels. For elements assigned to the $\ell$-th level, we will assign weights, either being $2^\ell$ or zero, in accordance with a carefully formulated distribution.

For each level $\ell \in [\kappa]$, we fix the size of labels $L_\ell$ {of level $\ell$} as $d/2^\ell$. Subsequently, we define $n$ as the cumulative sum of the sizes of these label sets, expressed as $n:=\sum_{\ell=1}^\kappa |L_\ell|$. In particular, we define 
label set $L_\ell$ as $\left \{\sum_{j=1}^{\ell - 1} \frac{d}{2^j} +1 , \dots, \sum_{j=1}^{\ell } \frac{d}{2^j} \right \}$. Note that the collection $\set{L_\ell}_{\ell \in [\kappa]}$ effectively forms a partition of the set $[n]$. 
Consequently, the set of vectors $\FF_2^{2d} \times L_\ell \subseteq \FF_2^{2d} \times  [n]$ is categorized as belonging to the level $\ell$. Having established this framework, we now proceed to outline the pairwise independent weight distribution along with a deterministic worst-case arrival order of elements which renders this distribution challenging for a gambler.

First, we describe the 
weight assignment procedure for elements in the matroid, with an approach akin to the active element distribution for the pairwise independent CRS problem, described in Procedure~\ref{proc:construct_crs_active_elements}.
More formally, for each level $\ell \in [\kappa]$, we begin by constructing a random matrix $\Sigma_\ell \in \FF_2^{d\times |L_\ell|}$ whose columns are pairwise linearly independent with probability one. Subsequent to this, we apply a linear transformation with a matrix $R$ whose entries are sampled independently from $r_{(i,j)} \sim \Ber(1/2)$. This step produces a new random matrix $X_\ell := R\Sigma_\ell$ for each level $\ell \in [\kappa]$.

Next, $X_\ell$ is converted into a subset $A_\ell$ of $\FF_2^{2d} \times L_\ell$ by independently applying Procedure~\ref{proc:construct_pw_set} across all levels. Each vector $\vec v^i$ in $A_\ell$ is assigned a weight of $2^\ell$. Notice that this process is well defined as $\set{L_\ell}_{\ell \in [\kappa]}$ forms a partition of $[n]$. As we will show in Lemma~\ref{lem:key_pwind_prophet}, weights sampled according to this process are pairwise independent, provided the matrix $\Sigma := \begin{bmatrix}
	\Sigma_1\ \Sigma_2\ \dots \Sigma_\kappa
\end{bmatrix}$ comprises pairwise linearly independent columns.

Regarding the deterministic order in the prophet inequality problem, we define an arbitrary fixed sequence $\lambda_\ell$ for the elements within each level. The worst-case order $\lambda$ for the matroid $\Ft^{2d} \times [n]$ is then composed by concatenating these individual level sequences in ascending order of levels. This is expressed as $\lambda: = \lambda_1 , \lambda_2 , \dots, \lambda_\kappa$. Crucially, this sequence is structured without foreknowledge of the weight realizations, ensuring that any algorithm addressing the prophet inequality will confront elements with nonzero weights in ascending order of their assigned weights.

\subsubsection{\texorpdfstring{Properties of $\Sigma$ matrices}{Properties of Sigma matrices}}
We encode the challenging instance of the prophet inequality problem via a careful selection of the matrices $\Sigma_1, \dots, \Sigma_\kappa$, each adhering to certain essential properties.
At high level, we designate a random subspace for each level $\ell$, derived from the span of a randomly chosen subset of principal basis vectors $B_\ell$. The matrices $\Sigma_\ell$ are then selected in a way that their columns are linearly independent subset within $\spn(B_\ell)$. For convenience, we think of $\Sigma_\ell$ interchangeably as a matrix as well as the set of it's columns, and write $e \in \Sigma_\ell$ for $e \in \FF_2^d$ if $e$ is a column of $\Sigma_\ell$. The following four properties define critical features of $\set{\Sigma_\ell}_{\ell \in [\kappa]}$ and $\set{B_\ell}_{\ell \in [\kappa]}$ which ensure that our resulting  $\vec w$ and $\lambda$ constitute a hard instance for the pairwise-independent prophet inequality problem as explained in Section~\ref{sec:technical-pw}.

\begin{enumerate}[label=(\roman*)]
	\item  $B_\kappa \subseteq \dots \subseteq B_2 \subseteq B_1 \subseteq \FF_2^d$ is a nested system with size $|B_\ell| = \frac{d}{2^{\ell-1}}$. \label{property:nested_B}
	\item  $\Sigma_\ell \in \FF_2^{d \times d/2^{\ell}}$ is a full column rank matrix with columns from $\spn(B_\ell)$.\label{property:almost_fullrank_sigma}
	\item  $\Pr[e \in \spn(B_{\ell'}) \mid \{e \in \Sigma_{\ell}\}, \Sigma_1, \dots, \Sigma_\ell, B_1, \ldots, B_\ell] = \frac{1}{2^{\ell'-\ell}}$ for all $e \in \FF_2^d$ and $\ell' \geq \ell$. \label{property:level_increase}
	\item  $\Sigma:=\begin{bmatrix}
		\Sigma_1\ \Sigma_2\ \dots \Sigma_{\kappa}
	\end{bmatrix}$ consists of distinct columns with probability $1$. \label{property:pw_indep_sigma}
\end{enumerate}

Let us revisit the roles of these properties, briefly touched upon in Section~\ref{sec:technical-pw}. We first set aside for a moment the requirement for the distribution to be pairwise independent. Consider a ``hard'' instance of a prophet inequality problem on $\FF_2^d$ where each column vector in some $\Sigma_\ell$ is assigned a weight of $2^\ell$ and the rest of the vectors are assigned  weight zero. Assume that the vectors with non-zero weight appear in  ascending order of weight, starting with vectors in $\Sigma_1, \Sigma_2, \dots, \Sigma_\kappa$, followed by vectors with weight zero.

A greedy offline algorithm, aiming to select the maximum weight independent set, would process the non-zero-weight vectors in reverse order from $\Sigma_\kappa$ to $\Sigma_1$. By utilizing the matroid exchange principle and Property~\ref{property:almost_fullrank_sigma}, the algorithm can secure a constant fraction of each set $\Sigma_\ell$. This strategy enables the prophet to achieve a total reward of $\Omega(d \cdot \log d)$.
 
Conversely,  a gambler must select a subset $T_\ell$ from the columns of $\Sigma_\ell$ without prior knowledge of the sets $\Sigma_{\ell'}$ or their corresponding weights, for all $\ell' > \ell$. Roughly speaking, based on Properties~\ref{property:nested_B} and~\ref{property:level_increase}, 
we can potentially show that if $|T_\ell| = c\cdot |\Sigma_\ell |$ for some constant $c$ then $T_\ell$ ends up spanning $c$-fraction of vectors from all the sets $\Sigma_{\ell'}$ for all $\ell' >\ell$. Hence, the gambler can only select $c$ fraction of vectors from $\frac 1 c$ many levels. This leads an $O(d)$ upper bound on the total reward of the gambler and rules out $\omega(1/\log d)$-competitive algorithm for the gambler.

Returning to the aspect of pairwise independence, Property~\ref{property:pw_indep_sigma} facilitates the conversion of this ``hard'' instance into one where weights are pairwise independent, while maintaining the problem's inherent difficulty. This step utilizes the tools developed in Section~\ref{sec:tool}.

\subsubsection{Formal description of hard instance}
\begin{figure}[t]
\begin{procedure}{Pairwise Independent Weight Assignments to Matroid $\FF_2^{2d} \times [n]$}\label{proc:constructing_weights_prophet_inequality}
	\textbf{Input: } Dimension of column vector $d$, Number of levels $\kappa$.
	
	\textbf{Output:} Weight assignment $\w$ and $\lambda$.

	\begin{itemize}
		\item Construction of random matrices with pairwise independent columns.
		\begin{enumerate}[label={(\arabic*)}]
			\item \label{step:construct-sigma-crs} Let $\Sigma_1, \Sigma_2, \dots \Sigma_\kappa$ be random matrices satisfying \ref{property:nested_B}, \ref{property:almost_fullrank_sigma}, \ref{property:level_increase} \ref{property:pw_indep_sigma}.
			\item \label{step:define_R} Let matrix $R\in \FF_2^{2d \times d}$ to be a random matrix with entries $r_{ij} \sim \Ber(1/2)$ independently. 
			\item Define $X_\ell = R \Sigma_\ell \in \FF_2^{2d \times d/2^\ell}$ for each $\ell \in [\kappa]$.
			\item Define $n := \sum_{\ell=1}^\kappa \frac{d}{2^\ell}.$
		\end{enumerate}
		\item Conversion from matrices to pairwise independent sets and weight assignment.
		\begin{enumerate}[label={(\arabic*)}]
			\setcounter{enumi}{4}
			\item Define labels of level $\ell$ as $L_\ell := \left \{\sum_{j=1}^{\ell - 1} \frac{d}{2^j} +1 , \dots, \sum_{j=1}^{\ell } \frac{d}{2^j} \right \} \subseteq [n]$.
			\item For each $\ell \in [\kappa]$, call Procedure~\ref{proc:construct_pw_set}  with $X_\ell$ and $L_\ell$  to obtain $A_\ell \leftarrow \MtS(X_\ell, L_\ell)$.
			\label{step:construct_A}
			\item 
			For each $\vec v^i \in \FF_2^{2d} \times L_\ell$, assign weight 
			$ \vec w(\vec v^i) = \begin{cases}
				2^\ell & \text{if } \vec v^i \in A_\ell,\\
				0 &\text{otherwise.}
			\end{cases} $
		\end{enumerate}
		\item Definition of a worst-case deterministic order.
		\begin{enumerate}[label={(\arabic*)}]
			\setcounter{enumi}{9}
			\item Let   $\lambda_\ell$ be  an arbitrary ordering of elements $\{\vec v^i: \vec v \in \mathbb F_2^{2d} \text{ and } i\in L_\ell  \} = \FF_2^{2d} \times L_\ell $.
			\item Concatenate  $\lambda_\ell$ in increasing order of $\ell$ to form the ordering $\lambda = \lambda_1,  \dots \lambda_\kappa$ of $\FF_2^{2d}\times [n]$.
		\end{enumerate}
	\end{itemize}
      \end{procedure}
    \end{figure}
At this stage, we provisionally accept the existence of a random sampling for the sets $\{B_\ell\}_{\ell \in [\kappa]}$ and matrices $\{\Sigma_\ell\}_{\ell \in [\kappa]}$, which adhere to the four specified properties in unison. We postpone the discussion of how this distribution is explicitly constructed to Section~\ref{sec:pi_sigma_construct}. Now, we are prepared to describe the detailed process for determining the weight assignments and the sequence $\lambda$ in Procedure~\ref{proc:constructing_weights_prophet_inequality}.
    
Next, we show that the weight assignment constructed by Procedure~\ref{proc:constructing_weights_prophet_inequality} is pairwise independent whose proof simply follows from Property~\ref{property:pw_indep_sigma} and Lemmas~\ref{lem:setConstruct} and ~\ref{lem:rotation}. 

\begin{restatable}{lemma}{lemKeyPwindProphet}\label{lem:key_pwind_prophet}
	Given any $\kappa \geq 4$ and $d = 2^{2 \cdot \kappa}$, let $n = \sum_{\ell=1}^\kappa \frac{d}{2^{\ell}}$ and $\w$ 
	be the random weight assignment of elements of $\FF_2^{2d} \times [n]$ sampled according to the Procedure~\ref{proc:constructing_weights_prophet_inequality}. Then, for any distinct pairs of labeled vectors $\vec v^i, \vec u^j \in \FF_2^{2d} \times [n]$, and $\ell, \ell' \in [\kappa]$ we have 
	$$\Pr[\w(\vec v^i)=2^\ell \wedge \w(\vec u^j)=2^{\ell'}] = \Pr[\w(\vec v^i) = 2^\ell] \cdot \Pr[\w(\vec u^j) = 2^{\ell'}].$$
\end{restatable}

\begin{proof}
	We consider three cases:
	\begin{itemize}
		\item \textbf{Case 1 ($i\notin L_\ell$ or $j \notin L_{\ell'}$)}: We observe that $\Pr[\w(\vec v^i)=2^\ell \wedge \w(\vec u^j)=2^{\ell'}] = \Pr[\w(\vec v^i) = 2^\ell] \cdot \Pr[\w(\vec u^j) = 2^{\ell'}] = 0$.
		\item \textbf{Case 2 ($\ell = \ell'$ and $i,j \in L_\ell$)}: Property~\ref{property:pw_indep_sigma} and Lemma~\ref{lem:setConstruct} ensures that the events $\{\vec v^i \in A_\ell\}$ and $\{\vec v^i \in A_\ell\}$ are independent.
\filbreak
                \item \textbf{Case 3 ($\ell \neq \ell'$ and $i \in L_\ell$ and $j \in L_{\ell'}$)}: We have marginal probabilities $\Pr[\vec v^i = 2^\ell] = \Pr[\vec u^j = 2^{\ell'}]=\frac 1 {2^d}$. Recall that $A_\ell$ sampled according to mixture of two distributions $\D_1$ and~$\D_2$ in $\MtS$ for each level $\ell$. We compute the joint probability as                
                  \begin{align*}
			\Pr[\vec v^i = &2^\ell \land \vec u^j = 2^{\ell'}]\\
				&= \Pr[\set{\vec x_i = \vec v} \land \set{\vec x_j = \vec u} \mid A_\ell \sim \D_1 \wedge A_{\ell'} \sim \D_1] \cdot \Pr[A_\ell \sim \D_1 \wedge A_{\ell'} \sim \D_1]\\
				&\qquad+ \Pr[\set{\vec x_i = \vec v} \land \set{\vec u^j \in A_{\ell'}} \mid A_\ell \sim \D_1 \land A_{\ell' } \sim \D_2] \cdot \Pr[A_\ell \sim \D_1 \land A_{\ell' } \sim \D_2]\\
			 	&\qquad+ \Pr[\set{\vec v^i \in A_\ell} \land \set{\vec x_j = \vec u} \mid A_{\ell} \sim \D_2 \land A_{\ell'} \sim \D_1] \cdot \Pr[A_{\ell} \sim \D_2 \land A_{\ell'} \sim \D_1]\\
				&\qquad+ \Pr[\set{\vec v^i \in A_\ell} \land \set{\vec u^j \in A_{\ell'}} \mid A_\ell \sim \D_2 \wedge A_{\ell'} \sim \D_2] \cdot \Pr[A_\ell \sim \D_2 \wedge A_{\ell'} \sim \D_2]\\
			&= \frac{1}{2^{2d}}\left(  1 - \frac{1}{2^d} \right)^2 + \frac{1}{2^{3d}}\left( 1- \frac{1}{2^d} \right)+ \frac{1}{2^{3d}}\left( 1- \frac{1}{2^d} \right) + \frac{1}{2^{4d}} = \frac 1 {2^{2d}}. \qedhere
		\end{align*}
	\end{itemize}

\end{proof}


\filbreak
\subsection{Upper Bounding the Approximation Ratio}
\label{sec:pi_upperbound}
In this section, we prove the following theorem.
\begin{theorem}
    \label{thm:prophet_main}
    For any integers $\kappa \geq 4$ and $d = 2^{2 \cdot \kappa}$, let  $n=\sum_{\ell=1}^\kappa \frac{d}{2^\ell}$ and $\M^{\times n}=(E^{\times n}, \I^n)$ where $\M=(E, \I)$ is the full binary matroid with rank $2d$, i.e., $E=\FF_2^{2d}$. Then, there is no $\omega(1/\log d)$-competitive pairwise independent prophet inequality algorithm for matroid $\M^{\times n}$.
\end{theorem}

In order to prove the above theorem, we show that when the weights $\vec w$ and the order $\lambda$ are determined according to Procedure~\ref{proc:constructing_weights_prophet_inequality}, the prophet is able to secure a reward of $\Omega(d \cdot \log d)$, where as any gambler is limited to obtaining a reward at most $O(d)$. To prove this claim, we first identify a high probability event which makes the problem challenging for the gambler. Let $\Ehard$ be the event when the following two events happen together:
\begin{enumerate}[label=(\arabic*)]
    \item The random matrix $R$ sampled in Step~\ref{step:define_R} of Procedure~\ref{proc:constructing_weights_prophet_inequality} has full column rank,
   	\item The set $A_\ell$ constructed in the Step~\ref{step:construct_A} of Procedure~\ref{proc:constructing_weights_prophet_inequality}, is sampled according to $\D_1$ in Procedure~\ref{proc:construct_pw_set} for all $\ell \in [\kappa]$ simultaneously.
\end{enumerate}
Notice that the two described events are mutually independent. The first event occurs with a probability of $1-\frac{1}{2^d}$, as $R \in \FF_q^{2d \times d}$. The second event occurs with a probability of $\left(1-\frac{1}{2^{2d}}\right)^\kappa \geq 1-\frac{\kappa}{2^{2d}}$. Together, these events imply that $\Ehard$ occurs with a probability of at least $1- \frac{\kappa+1}{2^d}$. 

Furthermore, when the latter event happens only labeled copies of columns in $X_\ell$ will have a non-zero weight for each $\ell \in [\kappa]$. Additionally, if the first event occurs, at most one labeled copy of each vector will have a non-zero weight. Conceptually, when $\Ehard$ occurs, one can imagine that $\vec w$ assigns non-zero weights to columns of $X_\ell$ and so vectors of $\FF_2^{2d}$, rather than their labeled copies, since there is only one labeled copy of each item with a non-zero weight.

Next, we demonstrate two lemmas that lower-bound the prophet's reward and upper-bound the any gambler's utility respectively, under the condition that $\Ehard$ happens.
\begin{lemma}[Lower-bound Prophet's Reward]
    \label{lem:prophet_lowerbound}
    Given $\kappa \geq 4, d = 2^{2 \cdot \kappa}, n=\sum_{\ell=1}^\kappa \frac{d}{2^\ell}$ and $\M = \FF_2^{2d}$, let $\vec w$ be the random weight assignment over elements of $\M^{\times n}$ determined according to Procedure~\ref{proc:constructing_weights_prophet_inequality}. Then,
        $$ \Ex[\text{reward of prophet} \mid \Ehard] \geq \frac{\kappa \cdot d}{10}. $$
\end{lemma}
\begin{lemma}[Upper-bound Gambler's Reward]
    \label{lem:gambler_upperbound}
    Given $\kappa \geq 4, d = 2^{2 \cdot \kappa}, n=\sum_{\ell=1}^\kappa \frac{d}{2^\ell}$ and $\M = \FF_2^{2d}$, let $\vec w$ be the random weight assignment over elements of $\M^{\times n}$ and $\lambda$ be the order of these elements determined according to Procedure~\ref{proc:constructing_weights_prophet_inequality}. Then,
    $$ \Ex[\text{reward of gambler} \mid \Ehard] \leq 2d. $$
\end{lemma}
Finally, we utilize these lemmas to prove the main result of this section.
\begin{proof}[Proof of Theorem~\ref{thm:prophet_main}]
    Due to Lemma~\ref{lem:prophet_lowerbound}, we have
    \begin{align}
        \Ex[\text{reward of prophet}] &\geq \Ex[\text{reward of prophet} \mid \Ehard] \cdot \Pr[\Ehard]\\
        &\geq \frac{\kappa \cdot d}{4} \cdot \left(1-\frac{\kappa+1}{2^d}\right) \geq \frac{\kappa \cdot d}{5} && (\kappa \geq 4)
    \end{align}
    and due to Lemma~\ref{lem:gambler_upperbound}, we have
    \begin{align*}
        \Ex[\text{reward of gambler}] &\leq \Ex[\text{reward of gambler} \mid \Ehard]  + 2^\kappa \cdot \rank(\M^{\times n}) \cdot \left(1-\Pr\left[\Ehard\right]\right)\\
        &\leq 2d + (d \cdot 2^\kappa) \cdot \frac{\kappa + 1}{2^d}\leq 2d + 1. && (\kappa \geq 4) 
    \end{align*}
    Therefore, combining two inequalities concludes the desired result.
    $$ \frac{\Ex[\text{reward of gambler}]}{\Ex[\text{reward of prophet}]} \leq \frac{10}{\kappa} = O\left(\frac{1}{\log d}\right). $$
\end{proof}
In the remaining part of this section, we condition on the event $\Ehard$ and prove the two lemmas.

\subsubsection{Proof of Lemma~\ref{lem:prophet_lowerbound}: Lower-bound for Prophet's Reward}

To compute a lower-bound for prophet's reward, or equivalently $\Ex[\rank_{\vec w}(\M^{\times n})]$, we will construct a feasible solution $S$ and show that its expected weight is at least $\frac{\kappa \cdot d}{2}$. Let $A_1, \dots A_\kappa$ be the random sets constructed in Step~\ref{step:construct_A} of Procedure~\ref{proc:constructing_weights_prophet_inequality}. 
Given that we have conditioned on $\Ehard$, each set $A_\ell$ contains $n_\ell = \frac{d}{2^\ell}$ vectors, and at most one labeled copy of a vector $\vec v \in \FF_2^{2d}$ appears in any of the sets $\set{A_\ell}_{\ell \in [\kappa]}$. 

Then, for each level $\ell$ in decreasing order, let $\tilde A_\ell$ be a maximal subset of $A_\ell$ such that $\bigcup_{i=\ell}^{\kappa} \tilde A_\ell$ is an independent set. We form the solution $S:=\bigcup_{\ell=1}^\kappa \tilde A_\ell$. Essentially, $S$ is the outcome of a greedy algorithm that processes elements in descending order of their weights.
By construction, $S$ constitutes a feasible set. The matroid exchange property then implies that
$$|\tilde A_\ell| = \rank\left( \bigcup_{i=\ell}^\kappa A_i \right) -  \rank\left(\bigcup_{i=\ell+1}^\kappa A_i\right).$$
Consequently, for any $1 \leq i \in \ceil{\kappa/2}$,
\begin{align*}
    |\tilde A_{2i-1}| + |\tilde A_{2i}| &= \rank\left(\bigcup_{j=2i-1}^\kappa A_j\right) - \rank\left(\bigcup_{j=2i+1}^\kappa A_j\right)\\
    &\geq \rank(A_{2i-1}) - \rank(B_{2i+1}) && \left(\bigcup_{j=2i+1}^\kappa A_j \subseteq \spn(B_{2i+1})\right)\\
    &= \frac{d}{2^{2i-1}} - \frac{d}{2^{2i}} 
    = \frac{1}{2} \cdot \frac{d}{2^{2i-1}}.
\end{align*}
Therefore, we can bound the weighted rank of the matroid as follows:
\begin{align*}
    \rank_{\vec w}(\M^{\times n}) 
    &\geq \vec w(S) = \sum_{\ell=1}^\kappa \vec w(\tilde A_\ell)\geq \sum_{i=1}^{\ceil{\kappa/2}} (|\tilde A_{2i-1}| + |\tilde A_{2i}|) \cdot 2^{2i-1}\\
    &\geq \sum_{i=1}^{\ceil{\kappa/2}} \frac{1}{2} \cdot \frac{d}{2^{2i-1}} \cdot 2^{2i-1}= \sum_{i=1}^{\ceil{\kappa/2}} \frac{d}{2} = \frac{\kappa \cdot d}{4}
\end{align*}
where the second inequality follows from the fact that $\vec w(\vec v^i) \geq 2^{2i-1}$ for any $\vec v^i \in \tilde A_{2i-1} \cup \tilde A_{2i}$. This concludes the proof of Lemma~\ref{lem:prophet_lowerbound}.


\subsubsection{Proof of Lemma~\ref{lem:gambler_upperbound}: Upper-bound for Gambler's Reward}

Let $S$ be the set of elements selected by the gambler and define $S_\ell$ as its subset which contains elements from level $\ell$, i.e., $S_\ell = S \cap \FF_q^{2d} \times L_\ell$, and let $S_{\leq \ell} = \bigcup_{i=1}^\ell S_i$. We assume without loss of generality that the gambler only selects elements with non-zero weight. 
Recall that $\lambda$ first presents all elements with first level labels $L_1$, then all elements with second level labels $L_2$, and so on. Thus, any algorithm visits non-zero-weight elements of the matroid in increasing order of weight. Therefore, $S_\ell$ is a random set which is independent of $B_{\ell+1}, B_{\ell+2}, \ldots, B_\kappa$ given the information revealed by the elements up to level $\ell$. This independence remains even when conditioned on $R$.

We define $\rho_\ell := |S_\ell|$ and $\gamma_\ell := |S_{\leq \ell} \cap \spn(RB_\ell \times [n])|$ for each $\ell \in [\kappa]$. To simplify the notation and reduce clutter, we will use a shorthand for $\spn(\cdot)$, where $\spn(V)$ will represent $\spn(V \times [n])$, effectively omitting labels for the ease of understanding. The following claim establishes a relationship between the two series $\set{\rho_\ell}_{\ell \in [\kappa]}$ and $\set{\gamma_\ell}_{\ell \in [\kappa]}$ in expectation.
\begin{claim}\label{claim:rho_gamma}
	$\Ex[\rho_\ell] \leq \Ex[\gamma_\ell] - \frac{\Ex[\gamma_{\ell-1}]}{2}.$
\end{claim}
Before we prove the claim, we utilize it to complete the proof of Lemma~\ref{lem:gambler_upperbound}. Remember that all expectations are conditioned on the event $\Ehard$.
\begin{align*}
		\Ex[\vec w(S)] &= \sum_{\ell=1}^{\kappa} \Ex[2^\ell \cdot |S_\ell|] \\
		&= \sum_{\ell=1}^{\kappa} \Ex[2^\ell \cdot \rho_\ell] \\
		&\leq 2 \cdot \Ex[\gamma_1] + \sum_{\ell=2}^\kappa 2^\ell \cdot \left(\Ex[\gamma_\ell] - \Ex\left[\frac{\gamma_{\ell-1}}{2}\right] \right) && (\text{Claim~\ref{claim:rho_gamma}})\\
		&= 2^\kappa \cdot \Ex[\gamma_\kappa] + \sum_{\ell=1}^{\kappa-1} \left(2^{\ell} \cdot \Ex[\gamma_\ell] - 2^{\ell+1} \cdot \Ex\left[\frac{\gamma_{\ell}}{2}\right]\right)\\
		&\leq 2^\kappa \cdot \rank(RB_\kappa) && (S_{\leq \kappa} \cap \spn(RB_\kappa) \subseteq \spn(RB_\kappa))\\
		&= 2^\kappa \cdot \frac{d}{2^{\kappa-1}} = 2d.
\end{align*}

We complete the discussion with the following proof of Claim~\ref{claim:rho_gamma}.
\begin{proof}[Proof of Claim~\ref{claim:rho_gamma}]
		First, we bound $\rho_\ell$ as follows.
	\begin{align*}
		\rho_\ell &\leq |S_{\leq\ell} \cap \spn(RB_\ell)| - |S_{\leq \ell-1} \cap \spn(RB_\ell)|= \gamma_\ell - |S_{\leq \ell-1} \cap \spn(RB_\ell)|. 
	\end{align*}
	Here, the inequality follows from the fact that $S_\ell \subseteq \spn(RB_\ell)$ and $S_{\leq \ell}$ is an independent set.

	We re-express the expectation of $|S_{\leq \ell-1} \cap \spn(RB_\ell)|$ to complete the proof.
	\begin{align*}
		\Ex [|S_{\leq \ell-1} \cap \spn(RB_{\ell})|]
		&= \sum_{\ell'=1}^{\ell-1} \Ex[|S_{\ell'} \cap \spn(RB_\ell)|] && (\text{disjoint})\\
		&= \sum_{\ell'=1}^{\ell-1} \Ex\left[ \sum_{e \in S_{\ell'}} \Pr[e \in \spn(RB_\ell) ] \right]\\
		&= \sum_{\ell'=1}^{\ell-1} \Ex\left[ \sum_{e \in S_{\ell'}} \frac{1}{2^{\ell-\ell'}}\right] && (\text{Property~\ref{property:level_increase}})\\
		&= \sum_{\ell'=1}^{\ell-1} \Ex\left[ \frac{1}{2} \cdot \sum_{e \in S_{\ell'}} \frac{1}{2^{(\ell-1)
		-\ell'}}\right]\\
		&= \sum_{\ell'=1}^{\ell-1} \Ex\left[ \frac{1}{2} \cdot \sum_{e \in S_{\ell'}} \Pr[e \in \spn(RB_{\ell-1}) ] \right] && (\text{Property~\ref{property:level_increase}}) \\
		&= \frac{1}{2} \cdot \sum_{\ell'=1}^{\ell-1} \Ex[|S_{\ell-1}  \cap \spn(RB_{\ell-1})|] \\
		&=\frac{1}{2} \cdot \Ex[\gamma_{\ell-1}]. && (\text{disjoint})
	\end{align*}
	Each of our two invocations Property~\ref{property:level_increase} uses the conditional independence of $S_{\ell'}$ of $B_\ell$  given the information revealed by elements at levels up to $\ell'$.
\end{proof}


\subsection{\texorpdfstring{Construction of $\Sigma$}{Construction of Sigma}}
\label{sec:pi_sigma_construct}

\begin{figure}
\begin{procedure}{Construction of Pairwise Linearly Independent Collection of Vectors $\Sigma$
	}\label{proc:constructing_sigma_prophet_inequality}
	\textbf{Input: } Dimension of column vector $d$, Number of levels $\kappa$.
	
	\begin{enumerate}[label={(\arabic*)}]
		\item \textbf{Base level $\ell=1$}:
		\begin{enumerate}
			\item $B_1 = \{\vec e_1 , \dots, \vec e_d\}$ be the principal basis of $\Ft^{d}$.
			\item $\mathbf P_1 = \{P_1(i) : i \in [d/2]\}$ be a partition of $B_1$ where $P_1(i)=\{\vec e_{2i-1}, \vec e_{2i}\}$.
		\end{enumerate}
		\item \textbf{Level $\ell>1$}:
		\begin{enumerate}
			\item Let $\tilde{\mathbf{P}}_{\ell-1}$ be a uniformly random half of $\mathbf{P}_{\ell-1}$, written $\tilde{\mathbf{P}}_{\ell-1}:=\{\tilde{P}_{\ell-1}(i) : i\in[d/2^{2\ell-2}]\}$.
			\item Define $P_{\ell}(i) = \tilde P_{\ell-1}(2i-1) \cup \tilde P_{\ell-1}(2i)$ for  $i \in \left[\frac{d}{2^{2\ell-1}}\right]$ and let $\mathbf{P}_{\ell} = \{P_{\ell}(i) : i \in [d/2^{2\ell-1}]\}$.
			\item Define $B_\ell = \bigcup_{i=1}^{d/2^{2\ell-1}} P_\ell(i)$, which we refer to as the alive basis vectors of level $\ell$.
		\end{enumerate}
		
		\item \label{step:vec-generation} 
		{For each $\ell \in [\kappa]$ and $j \in [d/2^{2\ell-1}]$, define $\Sigma_\ell(j) \in \Ft^{d \times 2^{\ell-1}}$ as follows. Let $\vec v'_1, \dots \vec v'_{2^{\ell}}$ be an arbitrary enumeration of $P_\ell(j)$, and let the $t$-th column of $\Sigma_\ell(j)$ be  $\sigma_t :=\sum_{i=t}^{t+2^{\ell-1}-1} \vec v'_i.$}
		\item Let $\Sigma_\ell = \begin{bmatrix}\Sigma_\ell(1)\ \Sigma_\ell(2)\ \dots\ \Sigma_\ell(d/2^{2\ell-1})\end{bmatrix} \in \FF_2^{d \times d/2^\ell}$. 
	\end{enumerate}	
	\textbf{Output: } $\left[\Sigma_1, \Sigma_2, \dots, \Sigma_\kappa\right]$.
      \end{procedure}
\end{figure}
This section is dedicated to presenting a method for constructing a random matrix $\Sigma$, along with a nested system of subspaces. Each subspace is defined as the span of a subset of principal basis vectors $B_\ell$ for every level $\ell \in [\kappa]$. The goal is to demonstrate that this construction adheres to Properties ~\ref{property:nested_B}, \ref{property:almost_fullrank_sigma}, \ref{property:level_increase}, and \ref{property:pw_indep_sigma} simultaneously. We will prove  the following lemma.

	\begin{lemma}\label{lem:propertiesofsigma}
 The matrices $\Sigma_1,\ldots,\Sigma_\kappa$ and sets $B_1\ldots,B_\kappa$ constructed by Procedure~\ref{proc:constructing_sigma_prophet_inequality} satisfy  
  Properties~\ref{property:nested_B}, \ref{property:almost_fullrank_sigma}, \ref{property:level_increase}, and \ref{property:pw_indep_sigma}.
\end{lemma}

To aid understanding, we now give a brief explanation of our construction. The procedure starts at the first level by defining the ``alive" basis vectors as $B_1 = \set{\vec e_1, \dots, \vec e_d}$, which are the principal basis vectors in $\FF_2^d$. It then organizes these vectors into consecutive pairs to create the set $\vec P_1$, a partition of $B_1$. We use $P_1(j)=\{e_{2j-1},e_{2j}\}$ to denote the $j$th pair in $\vec P_1$, for $j \in [d/2]$. We let $\Sigma_1(j) \in \FF_2^{d \times 1}$ be a matrix consisting of a single non-zero column, chosen as a linear combination of the two vectors in $P_1(j)$. These matrices are then concatenated to form the $d \times \frac{d}{2}$ matrix $\Sigma_1 = \left[ \Sigma_1(1),\ldots,\Sigma_1\left(\frac{d}{2}\right)\right]$ with full column rank.

In the second level, the procedure randomly selects half of the elements from $\vec P_1$ to form $\vec {\tilde P_1}$. For simplicity in our discussion, we will renumber the selected parts in $\vec{ \tilde P_1}$ as $\set{ \tilde P_1(i): i \in [d/4]}$. Following this, we combine consecutive parts in $\vec {\tilde P_1}$ to create the second level partition, formally denoted as $\vec {P_2}=\set{P_2(j): j \in [d/8]}$, where $P_2(j)=\tilde{P}_1(2j-1) \cup \tilde{P}_1(2j)$. The  basis vectors appearing in a part of $\vec{ P_2}$ are designated the ``alive" basis vectors at level two, and denoted by $B_2$. Note that the $B_2$ consists of half the vectors in $B_1$.
In a manner akin to the first level, vectors in each part $P_2(j)$ in $\vec P_2$ are linearly combined to create a  matrix $\Sigma_2(j) \in \FF_2^{d \times 2}$, with two linearly-independent columns lying within $\spn(P_{2}(j))$. Moreover, the procedure ensures that columns of each $\Sigma_2(j)$ are distinct from the columns of $\Sigma_1$. These matrices are then concatenated to form $\Sigma_2$.

\begin{figure}[t]
	\centering
	\begin{tikzpicture}[xscale=0.9]
    \def\marksegments{1,2,5,8,9,10,13,15}
    \def\MarkList{2,3,6,9,10,11,14,16}

    \edef\levelcord{0}
    \node[left] at (1,\levelcord) {$\textbf{P}_1$};
    \readlist*\activelist{1,1,1,1,1,1,1,1,1,1,1,1,1,1,1,1}

    \foreach \i in {1,2,...,16} {
        \ifnum\activelist[\i]=1
            \draw[fill=black, draw=black] (\i+0.3,\levelcord) circle (1.3mm);
            \draw[fill=black, draw=black] (\i+0.7,\levelcord) circle (1.3mm);
        \else
            \draw[fill=black, draw=black, opacity=0.2] (\i+0.3,\levelcord) circle (1.3mm);
            \draw[fill=black, draw=black, opacity=0.2] (\i+0.7,\levelcord) circle (1.3mm);
        \fi
    }

    \readlist*\boxcolor{0,1,0,1,0,1,0,1,0,1,0,1,0,1,0,1}
    \readlist*\begincord{1,2,3,4,5,6,7,8,9,10,11,12,13,14,15,16}
    \readlist*\endcord{2,3,4,5,6,7,8,9,10,11,12,13,14,15,16,17}
    \readlist*\croslist{0,1,1,0,0,1,0,1,0,1,1,1,0,0,1,0}
    \foreach \i in {1,2,...,16}{

        \ifnum\boxcolor[\i]=0
            \draw[fill=red, opacity=0.2] (\begincord[\i],\levelcord+-0.2) rectangle    (\endcord[\i],\levelcord+0.2);
        \else
            \draw[fill=blue, opacity=0.2] (\begincord[\i],\levelcord+-0.2) rectangle    (\endcord[\i],\levelcord+0.2);
        \fi
        
        \pgfmathparse{\begincord[\i]+\endcord[\i]}
        \pgfmathparse{\pgfmathresult/2}
        \xdef\midcord{\pgfmathresult}
        \node[below, font=\fontsize{6pt}{8pt}\selectfont] at (\midcord,\levelcord-0.32) {$P_1(\i)$};
    }

    \edef\levelcord{2}
    \node[left] at (1,\levelcord) {$\textbf{P}_2$};
    \readlist*\activelist{1,0,0,1,1,0,1,0,1,0,0,0,1,1,0,1}

    \foreach \i in {1,2,...,16} {
        \ifnum\activelist[\i]=1
            \draw[fill=black, draw=black] (\i+0.3,\levelcord) circle (1.3mm);
            \draw[fill=black, draw=black] (\i+0.7,\levelcord) circle (1.3mm);
        \else
            \draw[fill=gray, draw=gray, opacity=0.2] (\i+0.3,\levelcord) circle (1.3mm);
            \draw[fill=gray, draw=gray, opacity=0.2] (\i+0.7,\levelcord) circle (1.3mm);
        \fi
    }

    \readlist*\boxcolor{0,1,0,1}
    \readlist*\begincord{1,5,9,14}
    \readlist*\endcord{5,8,14,17}
    \readlist*\croslist{1,0,0,1}
    \foreach \i in {1,2,3,4}{

        \ifnum\boxcolor[\i]=0
            \draw[fill=red, opacity=0.2] (\begincord[\i],\levelcord+-0.2) rectangle    (\endcord[\i],\levelcord+0.2);
        \else
            \draw[fill=blue, opacity=0.2] (\begincord[\i],\levelcord+-0.2) rectangle    (\endcord[\i],\levelcord+0.2);
        \fi
        
        \pgfmathparse{\begincord[\i]+\endcord[\i]}
        \pgfmathparse{\pgfmathresult/2}
        \xdef\midcord{\pgfmathresult}
        \node[below, font=\fontsize{6pt}{8pt}\selectfont] at (\midcord,\levelcord-0.32) {$P_2(\i)$};
    }
    
    \edef\levelcord{4}
    \node[left] at (1,\levelcord) {$\textbf{P}_2$};
    \readlist*\activelist{0,0,0,0,1,0,1,0,1,0,0,0,1,0,0,0}

    \foreach \i in {1,2,...,16} {
        \ifnum\activelist[\i]=1
            \draw[fill=black, draw=black] (\i+0.3,\levelcord) circle (1.3mm);
            \draw[fill=black, draw=black] (\i+0.7,\levelcord) circle (1.3mm);
        \else
            \draw[fill=gray, draw=gray, opacity=0.2] (\i+0.3,\levelcord) circle (1.3mm);
            \draw[fill=gray, draw=gray, opacity=0.2] (\i+0.7,\levelcord) circle (1.3mm);
        \fi
    }

    \readlist*\boxcolor{0}
    \readlist*\begincord{5}
    \readlist*\endcord{14}
    \readlist*\croslist{0}
    \foreach \i in {1}{

        \ifnum\boxcolor[\i]=0
            \draw[fill=red, opacity=0.2] (\begincord[\i],\levelcord+-0.2) rectangle    (\endcord[\i],\levelcord+0.2);
        \else
            \draw[fill=blue, opacity=0.2] (\begincord[\i],\levelcord+-0.2) rectangle    (\endcord[\i],\levelcord+0.2);
        \fi
        
        \pgfmathparse{\begincord[\i]+\endcord[\i]}
        \pgfmathparse{\pgfmathresult/2}
        \xdef\midcord{\pgfmathresult}
        \node[below, font=\fontsize{6pt}{8pt}\selectfont] at (\midcord,\levelcord-0.30) {$P_3(\i)$};
    }
    
    \edef\levelcord{4}
    \node[left] at (1,\levelcord) {$\textbf{P}_3$};
    \readlist*\colormap  {2,2,2,2,0,2,0,2,0,2,2,2,0,2,2,2}
    \readlist*\activelist{1,0,0,1,1,0,1,0,1,0,0,0,1,1,0,1}

    \draw[->, >=latex, very thick] (1.5,0.5) -- (1.5,1.7);
    \draw[->, >=latex, very thick] (4.5,0.5) -- (4.5,1.7);
    \draw[->, >=latex, very thick] (5.5,0.5) -- (5.5,1.7);
    \draw[->, >=latex, very thick] (7.5,0.5) -- (7.5,1.7);
    \draw[->, >=latex, very thick] (9.5,0.5) -- (9.5,1.7);
    \draw[->, >=latex, very thick] (13.5,0.5) -- (13.5,1.7);
    \draw[->, >=latex, very thick] (14.5,0.5) -- (14.5,1.7);
    \draw[->, >=latex, very thick] (16.5,0.5) -- (16.5,1.7);

    \draw[->, >=latex, very thick] (6.5,2.5) -- (6.5,3.7);
    \draw[->, >=latex, very thick] (11.5,2.5) -- (11.5,3.7);

\end{tikzpicture}
	\caption{\textbf{A snapshot of the construction of $B_\ell$ and partition $\vec P_\ell$ for the initial $3$ levels.} In this illustration, vertically aligned dots in groups of three represent the same basis vector. At any given level $\ell$, vectors indicated in solid colors comprise the alive vectors $B_\ell$ of level $\ell$. 
	Furthermore, the colored boxes at each level $\ell$, along with their constituent alive basis vectors, represent the parts $P_\ell(i)$ of the partition of $\vec P_\ell$. Random parts $\vec {\tilde P_\ell}$ that survive through to the next level are indicated by upward arrows.}
	\label{fig:construction_sigma}
\end{figure}
For each subsequent level $\ell$, the process follows a similar pattern. Given $B_{\ell-1}, \vec P_{\ell - 1}$ and $\Sigma_{\ell - 1}$,
\begin{enumerate}
	\item We let $\tilde{\vec P}_{\ell - 1}$ be the random half of $\vec P_{\ell - 1}$ and then obtain $\vec P_\ell$ whose parts  are formed by merging two consecutive parts from $\tilde{\vec P}_{\ell - 1}$ after an arbitrary enumeration. We let $B_\ell$ consist of the vectors appearing in $\vec P_\ell$. Here, note that $|\vec P_{\ell }| = \frac{1}{4} |\vec P_{\ell-1}| = \frac{d}{2^{2\ell -1}}$. Moreover, $|P_\ell(j)| = 2^\ell $ for any $j\in [d/2^{2\ell -1}]$, and $|B_\ell| = \frac{d}{2^{\ell}-1}$.  We refer the reader to Figure~\ref{fig:construction_sigma} for a visualization of the construction of $B_\ell$, $\vec P_\ell$, and $\vec {\tilde P_{\ell}}$ for the initial three levels.
	\item For each part $P_\ell(j)$ in $\vec P_\ell$, a matrix $\Sigma_\ell(j) \in \FF_2^{d \times |P_\ell(j)|/2}$ with full column rank is constructed. Its columns lie within $\spn(P_\ell(j))$, and are distinct from the columns of $\Sigma_1,\dots, \Sigma_{\ell-1}$.
	\item The matrices $\Sigma_\ell(j)$ are concatenated to form $\Sigma_\ell$. \end{enumerate}

    Since each $\Sigma_\ell(j)$ has linearly-independent columns in $\spn(P_\ell(j))$, and $P_\ell$ partitions the linearly-independent set of vectors $B_\ell$, it follows that each $\Sigma_\ell$ has full column rank with its columns in $\spn(B_\ell)$.

\begin{figure}[t]
	\begin{align*}
		P_3(j) &= \begin{Bmatrix} &\color{white}\sigma_1 = \color{black}&v_1',\color{white},\color{black} &v_2',\color{white},\color{black} &v_3',\color{white},\color{black} &v_4',\color{white},\color{black} &v_5',\color{white},\color{black} &v_6',\color{white},\color{black} &v_7',\color{white},\color{black} &v_8' \end{Bmatrix}\\
		\Sigma_3(j) &= 
		\begin{bmatrix} 
			&\sigma_1 = &v_1' + &v_2' + &v_3' + &v_4' \color{white}+\color{black}& & &\\
			&\sigma_2 = & &v_2' + &v_3' + &v_4' + &v_5' \color{white}+\color{black} & &\\
			&\sigma_3 = & & &v_3' + &v_4' + &v_5' + &v_6' \color{white}+\color{black} &\\
			&\sigma_4 = & & & &v_4' + &v_5' + &v_6' + &v_7'\color{white}aaaaa\color{black}
		\end{bmatrix}^T
	\end{align*}
	\caption{\textbf{Construction of $\Sigma_3(j)$ matrix}.}
	\label{fig:vec-generation}
\end{figure}    

Our construction also guarantees that the columns of all the matrices $\Sigma_\ell$, across all the different levels $\ell$,  are distinct and therefore pairwise linearly-independent in $\FF_2$.  Each column of $\Sigma_\ell(j)$ is formulated as a sum of distinct subsets (which may overlap) of vectors from $P_\ell(j)$, with each subset containing $2^{\ell-1}$ vectors. Given that the number of columns in $\Sigma_\ell(j)$ is half the size  $P_\ell(j)$, such a construction is possible. In detail, given $P_\ell(j) = \{\vec v_1' , \dots, \vec v'_{2^\ell}\}$, the $t$-th column of $\Sigma_\ell(j)$ is defined as: $$\sigma_t :=\sum_{i=t}^{t+2^{\ell-1}-1} \vec v'_i.$$
As an illustrative example, we depict the construction of $\Sigma_3(j)$ for  $j\in [d/2^{2\ell -1}]$ in Figure~\ref{fig:vec-generation}.	

The final matrix $\Sigma$ is the concatenation of the matrices across all levels, namely $\Sigma:=\begin{bmatrix} \Sigma_1\ \Sigma_2\ \dots\ \Sigma_\kappa \end{bmatrix}$.
The precise construction is described in Procedure \ref{proc:constructing_sigma_prophet_inequality}.

We now present two essential observations about our construction, which are vital for the proof of Lemma~\ref{lem:propertiesofsigma}.

\begin{observation}\label{obs:blocks-survives}
For each level $\ell \in [\kappa]$ and part $P_\ell(j) \in \vec P_\ell$, there exists level $\ell'$ such that $\ell \leq \ell' \leq \kappa$ and
\begin{enumerate}
    \item $P_\ell(j) \subseteq B_i$ for all $i \leq \ell'$,
    \item $P_\ell(j) \cap B_i = \emptyset$ for all $i >\ell'$.
\end{enumerate}

\end{observation}
\begin{proof}
	Considering a specific part $P_\ell(j)$ within $\vec P_{\ell}$, it is constructed such that either it becomes part of a subsequent part $P_{\ell+1}(j')$ in $\vec P_{\ell+1}$, or it does not intersect with any part in $\vec P_{\ell+1}$. This pattern holds for any two consecutive levels. Therefore, for any two levels $\ell < \ell'$ and a given part $P_{\ell}(j)$, it is either entirely contained within a part $P_{\ell'}(j')$ in $\vec P_{\ell'}$, or it shares no common elements with any part in $\vec P_{\ell'}$. Since $\vec P_{\ell'}$ is a partition of $B_{\ell'}$, this implies that $P_\ell(j)$ is either a subset of $B_{\ell'}$ or has no overlap with $B_{\ell'}$. 
	Since $\set{B_i}_{i \in [\kappa]}$ form a nested system, i.e., $B_i \subseteq B_{i+1}$ for each $i \in [\kappa-1]$, the proof is complete.	
\end{proof}

The following observation directly follows from the fact that each part survives to the next level with probability $1/2$ independently.
\begin{observation}\label{obsn:probability-of-survival}
	For any level $\ell \in [\kappa-1]$ and part $P_\ell(j) \in \vec P_\ell$ for $j\in [d/2^{2\ell-1}]$,
	\begin{equation*}
		\Pr[P_\ell(j) \subseteq B_{\ell+1} \mid P_\ell(j) \subseteq B_\ell] = \frac 1 2.
	\end{equation*}
\end{observation}

We are now ready to prove the Lemma~\ref{lem:propertiesofsigma}. 
\begin{proof}[Proof of Lemma~\ref{lem:propertiesofsigma}]
	
	We will demonstrate separately that each property is met by the construction:
	\begin{enumerate}[label=(\roman*)]
		\item At each level, precisely half of the ``alive'' basis vectors continue to the next level. Initially, at the first level, there are $d$ alive basis vectors, i.e., $|B_1| = d$. We conclude Property~\ref{property:nested_B}.
		\item By construction, each column of $\Sigma_\ell$ is a sum of standard basis vectors in $B_\ell$. It remains to show linear independence. Recall that $\Sigma_\ell = [ \Sigma_\ell(1), \ldots, \Sigma_\ell(d/2^{2 \ell - 1}) ]$, with $\Sigma_\ell(j) \sse \spn(P_\ell(j))$. Also recall that $\vec P_\ell$ is a partition of $B_\ell$. Consider an arbitrary column $e$ of $\Sigma_{\ell}$, corresponding to the $t$th column of a particular $\Sigma_\ell(j)$. From Step~\ref{step:vec-generation} of Procedure~\,\ref{proc:constructing_sigma_prophet_inequality}, $e= \sum_{i=t}^{s} \vec v'_i$ for $s=t+2^{\ell -1} - 1$, where $\vec v'_1, \dots \vec v'_{2^{\ell}}$ is a fixed enumeration of $P_\ell(j)$. All columns of $\Sigma_\ell$ to the left of $e$, whether in $\Sigma_\ell(j)$ or otherwise, are orthogonal to $v'_{s}$, whereas $e$ is not. Therefore, $e$ is linearly independent of all columns to its left. Applying this argument inductively, we conclude that $\Sigma_\ell$ has full column rank. Therefore, our construction satisfies Property~\ref{property:almost_fullrank_sigma}.

		\item Fix two levels $\ell , \ell' \in [\kappa]$ such that $\ell ' > \ell$. Let  $\vec e \in \Sigma_\ell$ be an arbitrary column of $\Sigma_\ell$. Let $j$ be such that $\vec e$ is a column of $\Sigma_\ell(j)$. By construction, we have that $\vec e \in \spn(P_\ell(j))$. Then,
		\begin{align*}
			\Pr[e \in \spn(B_{\ell'}) &\mid \{e \in \Sigma_{\ell}(j)\} , \Sigma_1,\dots, \Sigma_\ell, B_1,\ldots,B_\ell]\\
			&= \Pr[P_\ell(j) \subseteq \spn(B_{\ell'})  \mid   \Sigma_1,\dots, \Sigma_\ell, B_1,\ldots,B_\ell]\\
			&= \prod_{i = \ell +1}^{\ell'} \Pr \left [P_\ell(j) \subseteq \spn(B_{i})  \mid P_{\ell}(j) \subseteq \spn(B_{i-1}), \Sigma_1,\ldots,\Sigma_\ell, B_1,\ldots,B_\ell
			\right]\\
			&=\prod_{i = \ell +1}^{\ell'} \Pr \left [P_\ell(j) \subseteq \spn(B_{i})  \mid
			P_{\ell}(j) \subseteq \spn(B_{i-1})
			\right]\\
			&= \left( \frac{1}{2} \right)^{\ell' - \ell }.
		\end{align*}
		Above, the first equality follows from Observation~\ref{obs:blocks-survives}. The second one follows from the fact that $\set{B_\ell}_{\ell \in [\kappa]}$ forms a nested system. The third equality is because $\Sigma_1,\ldots,\Sigma_\ell$ and $B_1,\ldots,B_\ell$ are functions of $\vec P_1,\ldots,\vec P_\ell$, conditionally independent of what transpires at higher levels. 
        The last equality follows from Observation~\ref{obsn:probability-of-survival}. Therefore, the construction satisfies Property~\ref{property:level_increase}.
		\item Each $\Sigma_\ell$ has full column rank, and therefore it's columns are distinct. It remains to compare columns $\sigma \in \Sigma_\ell$ and $\sigma' \in \Sigma_{\ell'}$ for $\ell \neq \ell'$. The vector $\sigma$ is the sum of $2^{\ell-1}$ principal basis vectors, while $\sigma'$ is the sum of $2^{\ell'-1}$ principal basis vectors. It follows that $\sigma$ and $\sigma'$ are distinct. This confirms Property~\ref{property:pw_indep_sigma}.
	\end{enumerate}	
\end{proof}

\subsection{Optimal Pairwise Independent Matroid Prophet Inequality}

In this section, we present an optimal (up to a constant factor) algorithm for the pairwise independent matroid prophet inequality problem against the almighty adversary. It is important to note that although our algorithm is designed to perform against the strongest adversary, the almighty adversary; our impossibility result holds against the weakest adversary, namely the oblivious adversary. 

The following theorem is the main result of the section.
\begin{theorem}\label{thm:log rank PwI PI}
	There exists an $\Omega \left(\frac{1}{\log\rank}\right)$-competitive algorithm for the pairwise independent prophet inequality problem against the almighty adversary for any given matroid $\M=(E, \I)$. Here, $\rank$ is a shorthand notation for $\rank(\M)$.
\end{theorem}

The algorithm defines the strategy for the gambler by dividing elements into \emph{weight buckets}. Let $\M=(E,\I)$ be a matroid and $\w\sim \D$ be a pairwise independent weight distribution. We use $\opt$ to denote the expected value of $\rank_{\w}(\M)$ under $\D$. To proceed, we set $k$ as the ceiling of $\log (8 \rank)$ and introduce $k+2$ buckets: $B_0, B_1, \dots B_k$, along with $B_\infty$. The buckets are defined as follows:
\begin{align*}
	B_0 = \left[0, \frac{\opt}{2\rank}\right),
	B_1 = \left[\frac{\opt}{2\rank}, \frac{2 \opt}{2\rank} \right), \dots , B_k = \left[ \cdot \frac{2^{k-1}  \opt}{2\rank},  \frac{2^k  \opt}{2\rank}\right),
	B_{\infty} = \left[ \frac{2^k \opt}{2\rank}, \infty \right).
\end{align*}

Given any draw of the weights $w \sim \mathcal D$, we partition the elements into random sets $E_0, E_1, \dots E_k$ and $E_\infty$ based on their realization. We define $E_i = \{e \in E: \vec w(e) \in B_i\}$ for any $i \in \{0, \dots k, \infty\}$, which represents the elements whose weight lies in bucket $B_i$. It is important to note that the sets $E_i$ for $i \in \{0,\dots, k, \infty \}$ are random. 

We define the expected optimal reward from bucket $B_i$ for any $i\in \{0,1, \dots, k, \infty \}$ as 
$$\opt(B_i) = \Ex\left[ \max_{S \subseteq E_i, S \in \I} \rank_{\vec w}(S) \right].$$ This represents the maximum expected weighted rank of elements belonging to bucket $B_i$. We further define $$B^* = \max_{i\in \{1,\dots, k,\infty\}} \opt(B_i) .$$
	Note that we ignore  bucket $B_0$ because every element $e\in E_0$ has weight $\vec w(e) \leq \frac{\opt}{2 \cdot \rank}$, and therefore $\opt(B_0) \leq \frac{\opt}{2}$. Next, we aim to upper-bound the expected reward of the prophet by considering the total expected optimal rewards from each bucket. By invoking the fact that $\opt(B_0)\leq \frac \opt 2$, we obtain:

\begin{equation}\label{eq:bound_opt}
	\opt \leq \sum_{i=0}^{k} \opt(B_i) + \opt(B_\infty) \leq 2 \cdot \left(\sum_{i=1}^{k} \opt(B_i) + \opt(B_\infty)\right).
\end{equation}
Based on these observations, we define Algorithm~\ref{alg:pi_algorithm}.
\begin{algorithm}[t]
	\caption{Pairwise Independent Matroid Prophet Inequality Algorithm.}\label{alg:logrankPI}
	\label{alg:pi_algorithm}
	\textbf{Input: }Matroid $\mathcal M = (E,\I)$ and pairwise independent joint weight distribution $\w\sim \D$.\\
	\textbf{Output:} $S$.
	\begin{algorithmic}
		\State Compute $\opt$ and buckets $B_0,\dots, B_{k+1}$
		\State Let $B^* = \argmax_{B_i \in \{B_1,\dots, B_{k}, B_\infty\}} \opt (B_i)$
		\State Set $S\gets \emptyset$
		\While{Visit elements $e \in E$ in given order}
		\If{$\w(e)\in B^*$ and $S\cup \{e\} \in \I$}
		\State Update $S\gets S\cup \{e\}$
		\EndIf
		\EndWhile
	\end{algorithmic}
\end{algorithm}

	To prove the competitive ratio of the algorithm, we consider two cases: (i) when $B^* = B_i$ for some $i \in [k]$, and (ii) when $B^* = B_\infty$. First, we demonstrate that within each bucket $B_i$ (where $i \in \{1, \dots, k\}$), the greedy algorithm obtains a substantial portion of the optimal solution. This is because the weights of items in each bucket differ by at most a factor of $2$

\begin{lemma}\label{lem:finite_weights}
	Let $S$ be the output of Algorithm~\ref{alg:pi_algorithm}. If  $B^* \in \{B_1, \dots B_k\}$ then
	$$\Ex[\w(S) ] \geq \frac{1}{2} \cdot \Ex[\rank_{\w}(B^*) ].$$
\end{lemma}
\begin{proof}
	Let $B_i=B^*$ for some $i \in [k]$ and $\ell = 2^{i-1} \cdot \frac{\opt}{\rank^2}$. Then, observe that $\ell \leq \w(e) \leq 2 \cdot \ell$ for all $e \in B_i$ and $\rank_{\w}(E_i) \leq 2\cdot \ell \cdot \rank(E_i)$. Since the greedy algorithm guarantees an independent set $S$ such that $|S|=\rank(E^*)$ with probability $1$, we have
	$$ \rank_{\w}(S) \geq \ell \cdot |S| = \ell \cdot \rank(E_i) \geq \frac{1}{2} \cdot \rank_{\w}(E_i).$$
	which completes the proof.
\end{proof}

Next, we consider the scenario when $B^*=B_\infty$.
\begin{lemma}\label{lem:infinite_wegihts}
	Let $S$ be the output of Algorithm~\ref{alg:pi_algorithm}. If  $B^* = B_\infty$ then
	$$\Ex[\w(S) ] \geq \frac{1}{2} \cdot \Ex[\rank_{\w}(B^*) ].$$
\end{lemma}
\begin{proof}
	First of all, observe that $\vec w(e) \geq \frac{2^k \cdot \opt}{2} \geq 4\cdot \opt$ as $k = \ceil{\log (8 \rank)}$. Therefore,
	$$ \Pr[|E_\infty| \geq 1] \cdot 4 \cdot \opt \leq \Ex[\rank_w(E_\infty)] \leq \opt, $$
	and so $\Pr[E_\infty \neq \emptyset] \leq \frac{1}{4}.$ 
	As we have pairwise independent random weights, we use Lemma~\ref{lem:local-lemma-type} to obtain
	$$ \Pr[E_\infty \neq \emptyset] = \Pr\left[ \bigvee_{e\in E} \{e\in E_\infty\} \right] \geq \frac{\sum_{e \in E} \Pr[e \in E_\infty]}{1 + \sum_{e \in E} \Pr[e \in E_\infty]}. $$
	Rearranging terms implies that $\sum_{e \in E} \Pr[e \in E_\infty] \leq \frac{\Pr[E_\infty \neq \emptyset]}{1-\Pr[E_\infty \neq \emptyset]} \leq \frac{1}{3}.$ Thus, by using pairwise independence we obtain that for any element $e \in E$ and value $v \in B_\infty$, 
	$$ \Pr[|E_\infty|>1 \mid \vec w(e) = v] \leq \sum_{f \in E \setminus \set{e}} \Pr[f \in E_\infty \mid \vec w(e)=v]  =  \sum_{f \in E\setminus \set{e}} \Pr[f \in E_\infty] \leq \frac{1}{3}.$$
	Above, the first inequality is due to union bound. The equality follows from pairwise independence of the weights. This imply that for any value $v\in B_\infty$, we have $\Pr[E_\infty = \set{e} \mid \vec w(e) = v] \geq \frac{2}{3}$. Finally, we compute the expected weight of $S$ as 

	\begin{align*}
		\Ex[\vec w(S)] &= \sum_{e \in E} \Ex[\vec w(e) \cdot \indicator[e \in S]]\\
			&\geq \sum_{e \in E} \Ex[\vec w(e) \cdot \indicator[ |E_\infty| = 1] \cdot \indicator[ \vec w(e) \in B_\infty]]\\
			&= \sum_{e \in E} \Ex[ \Pr[ |E_\infty| = 1 \mid \vec w(e) \wedge \vec w(e) \in B_\infty] \cdot \vec w(e) \cdot \indicator[ \vec w(e) \in B_\infty]]\\
			&\geq \sum_{e \in E} \Ex\left[ \frac{2}{3} \cdot \vec w(e) \cdot \indicator[ \vec w(e) \in B_\infty]\right]\\
			&= \frac{2}{3} \cdot \Ex[ \vec w(B_\infty)] \geq \frac{2}{3} \cdot \Ex[\rank_{\vec w}(B_\infty)] \qedhere
	\end{align*}

\end{proof}

We now complete the proof of Theorem~\ref{thm:log rank PwI PI}.
\begin{proof}
	Let $B^*:= \argmax_{B_i \in \set{B_1, \dots B_k, B_\infty}} \opt(B_i)$ as defined in Algorithm~\ref{alg:logrankPI}. We observe that 
	$$ \Ex[\rank_w(E^*)] = \max_{i \in \set{1, \dots k, \infty}} \opt(B_i) \geq \frac{1}{k+1} \cdot \sum_{i \in \set{1, \dots, k, \infty}} \opt(B_i) \geq \frac{1}{k+1} \cdot \frac{1}{2} \cdot \opt. $$
	By Lemma~\ref{lem:finite_weights} and Lemma~\ref{lem:infinite_wegihts}, we know that $\Ex[w(S)] \geq \frac{1}{2} \cdot \Ex[\rank_w(E^*)] = \frac{1}{2} \cdot \opt(B^*)$. Thus, $\Ex[w(S)] \geq \frac{1}{4(k+1)} \cdot \opt$ completes the proof as $k = \Theta(\log \rank)$. 
	Since our algorithm is deterministic and the probabilistic guarantees are valid regardless of the arrival order, the approximation guarantee holds even against the almighty adversary. 
\end{proof}


\section{The Partition Property and its Implications}\label{sec:partition-property and Pw Selection}

In this section, our focus is on matroids that exhibit a \emph{constant partition property}. We demonstrate that such matroids admit constant factor guarantees for both pairwise-independent contention resolution and prophet inequalities. We also state structural implications of these results for the partition property.

{A simple partition matroid is the disjoint union of rank one matroids, as defined in Section \ref{sec:matroid-prelims}. We use the following definition of the partition property from \cite{babaioff2009secretary}.}

\filbreak
\begin{definition}[Partition Property] \label{def:alphapartition}
  {We say a matroid $\M=(E,\I)$ satisfies the $\alpha$-partition property for $\alpha \in (0,1]$ if there exists a random simple partition matroid $\M' = (E',\I')$ satisfying
	\begin{enumerate}
        \item $E' \sse E$ and $\I' \subseteq \I$, and
        \item $ \rank_{\w} (\M) \geq \E_{\M'} \left[\rank_{\w}(\M') \right] \geq \alpha \cdot \rank_{\w}(\M)$ for every nonnegative weight vector $\w$.
	\end{enumerate}}
\end{definition}
It is known that many classes of matroids that are frequently encountered in discrete optimization satisfy an $\alpha$-partition property with $\alpha = O(1)$ \cite{babaioff2009secretary,soto2013matroid}.

We show that if a matroid $\M$ satisfies the $\alpha$ partition property, then it admits a pairwise independent (i) $\frac{\alpha}{3}$-competitive prophet inequality against the almighty adversary, (ii) $\frac{\alpha}{1.299} \cdot (1-1/e)$-balanced offline CRS, and (iii) $\frac{\alpha^2}{3.897} \cdot (1-1/e)$-balanced OCRS against the almighty adversary. These are proven in Theorem~\ref{thm:PIwithPartition},~\ref{thm:crs-partition}, and \ref{thm:ocrs-partition} respectively. Essentially, we demonstrate these results by reducing the stochastic selection problems with pairwise independent priors for such matroids to those of one-uniform matroids, utilizing the partition property. In doing so, we assume our algorithms can sample a simple partition matroid $\M'$ satisfying the conditions of Definition~\ref{def:alphapartition}.  We summarize the corollaries derived from our results in Table~\ref{tab:partition-CRS-PI}.

	\begin{table}
	\centering
	\begin{tabular}{ | M{2.6cm} || M{3 cm} | M{3 cm} |  M {2 cm} | M{4cm} |} 
		\hline
		\textbf{Constraint} &\textbf{Offline CRS} & \textbf{OCRS} & \textbf{Prophet Inequality} & \textbf{Notes \& References}\\
		\hline
		\hline
		Partition Matroid & $\frac {1}{1.299} \cdot \left( 1 - \frac 1 e\right) $ & $\frac {1}{1.299} \cdot \left( 1 - \frac 1 e\right) \cdot \frac{1}{3}$& $\frac 1 3$& Theorem~\ref{thm:PIwithPartition}~\ref{thm:crs-partition}~\ref{thm:ocrs-partition} \\
		\hline
		Graphic Matroid & $\frac {1}{1.299} \cdot \left( 1 - \frac 1 e\right) \cdot \frac 1 2$ & $\frac {1}{1.299} \cdot \left( 1 - \frac 1 e\right) \cdot \frac{1}{12}$& $ \frac{1}{6} $ & Theorem~\ref{thm:PIwithPartition}~\ref{thm:crs-partition}~\ref{thm:ocrs-partition} and \cite{babaioff2009secretary}\\
		\hline
		Co-Graphic Matroid & $\frac {1}{1.299} \cdot \left( 1 - \frac 1 e\right) \cdot \frac 1 3$ & $\frac {1}{1.299} \cdot \left( 1 - \frac 1 e\right) \cdot \frac 1 {27}$ & $\frac{1}{9} $ & Theorem~\ref{thm:PIwithPartition}~\ref{thm:crs-partition}~\ref{thm:ocrs-partition} and \cite{soto2013matroid}\\
		\hline
		Laminar Matroid & $\frac {1}{1.299} \cdot \left( 1 - \frac 1 e\right) \cdot \frac{1}{3\sqrt 3}$ &$\frac {1}{1.299} \cdot \left( 1 - \frac 1 e\right) \cdot \frac{1}{3^4}$ & $\frac{1}{9\sqrt 3}$ &Theorem~\ref{thm:PIwithPartition}~\ref{thm:crs-partition}~\ref{thm:ocrs-partition} and \cite{jaillet2013advances}\\
		\hline 
		Low Density Matroid &$\frac {1}{1.299} \cdot \left( 1 - \frac 1 e\right) \cdot\frac{1}{2\gamma }$&$\frac {1}{1.299} \cdot \left( 1 - \frac 1 e\right) \cdot\frac{1}{12\gamma^2 }$ & $\frac{1}{6\gamma }$&Theorem~\ref{thm:PIwithPartition}~\ref{thm:crs-partition}~\ref{thm:ocrs-partition} and \cite{soto2013matroid},  $\gamma = \max_{S\subseteq E} \frac{|S|}{\rank(S)}$ \\
		\hline 
		Column $k$ Sparse Matroid &$\frac {1}{1.299} \cdot \left( 1 - \frac 1 e\right) \cdot\frac{1}{2k}$& $\frac {1}{1.299} \cdot \left( 1 - \frac 1 e\right) \cdot\frac{1}{12k^2}$   &$\frac{1}{6k}$&Theorem~\ref{thm:PIwithPartition}~\ref{thm:crs-partition}~\ref{thm:ocrs-partition} and \cite{soto2013matroid}\\
		\hline
	\end{tabular}
	\caption{\label{tab:partition-CRS-PI} Summary of our results for matroids that satisfy the partition property.}
\end{table}

\subsection{Prophet Inequalities}
\label{sec:partition_pi}
	We now exhibit an $\frac{\alpha}{3}$-competitive algorithm for pairwise-independent prophet inequalities on matroids satisfying the $\alpha$-partition property. First, we recall a result from \cite{pi-uniform-prophet} which shows the existence of a $\frac 1 3$-competitive pairwise independent prophet inequality for rank one matroids. The analysis of their threshold-based algorithm can easily be seen to hold even against the almighty online adversary.
\begin{theorem}[\cite{pi-uniform-prophet}]
\label{thm:1-uniform-prophet}
    Given a rank one matroid over elements $E$ and a pairwise independent value distribution $\D \in \pwset$, there exists a $\frac{1}{3}$-competitive prophet inequality algorithm against the almighty adversary. 
\end{theorem}

Our algorithm samples a partition matroid $\M'$ as in Definition~\,\ref{def:alphapartition}, then applies the $\frac{1}{3}$-competitive prophet inequality to each part of $\M'$ separately. This is shown in Algorithm~\ref{alg:pi_partition_algorithm}.

\begin{algorithm}[t]
	\caption{Partition-Based Algorithm for Pairwise Independent Matroid Prophet Inequalities.}
	\label{alg:pi_partition_algorithm}
        
	\textbf{Input: } A matroid $\mathcal M = (E,\I)$ satisfying the 
	$\alpha$ partition property,  distribution $\D \in \pwset$, black-box access to algorithm $\A$ for single-choice pairwise-independent prophet inequalities 
	\begin{algorithmic}
		\State Let $\M'$ be a random partition matroid $\alpha$-approximating $\M$ as in Definition\,\ref{def:alphapartition}, and let $P_1,\ldots, P_r$ be its parts.
		\State Separately for each part $P_i$, invoke $\A$ for the rank one matroid on  $P_i$ using the restriction of $\D$ to $P_i$, and let $S_i$ be its output.
	\end{algorithmic}
	\textbf{Output: }Set $S \leftarrow \bigcup_{i=1}^r S_i$ .
\end{algorithm}

\begin{theorem}\label{thm:PIwithPartition}
{For matroids satisfying the $\alpha$-partition property for some $\alpha \in (0,1]$, there is an $\frac{\alpha}{3}$-competitive pairwise-independent prophet inequality against the almighty adversary.}
\end{theorem}

\begin{proof}
{We invoke Algorithm~\ref{alg:pi_partition_algorithm} with black box access to  a $\frac{1}{3}$-competitive prophet inequality algorithm $\A$ for rank one matroids and pairwise independent distributions as in Theorem~\,\ref{thm:1-uniform-prophet}. Let $\M$, $\D$, $\M'$, $\set{P_i}_{i=1}^r$ and $S= \union_i S_i$  be as in Algorithm~\ref{alg:pi_partition_algorithm},  and let $\w \sim \D$ be the realized stochastic weights. We have the following guarantee on the weight of the output $S$ conditioned on $\M'$.
	\begin{align}\label{eq:using_one_uniform_alg}
		\Ex[\w(S) \mid \M']
		&= \sum_{i=1}^r \Ex[\w(S_i) \mid \M'] && (\text{Linearity of expectation}) \notag\\
		&\geq \sum_{i=1}^r \Ex\left[\frac{1}{3} \cdot \rank_{\w}^{\M'}(P_i) \mid \M' \right] && (\text{Theorem~\ref{thm:1-uniform-prophet}})\notag\\
		&= \frac{1}{3} \cdot \Ex\left[\sum_{i=1}^r  \rank_{\w}^{\M'}(P_i) \mid \M' \right] && (\text{Linearity of expectation})\notag\\
		&= \frac{1}{3} \cdot \Ex\left[\rank_{\w}(\M') \mid \M' \right] &&(\text{$\M'$ is a partition matroid})
	\end{align}}
Next, taking  expectations over $\M'$, we have 
	\begin{align*}
		\Ex[w(S)] &= \Ex\left[ \Ex[\w(S) \mid \M'] \right]\\
		&\geq \Ex\left[ \frac{1}{3} \cdot \Ex\left[\rank_{\w}(\M') \mid \M' \right] \right] && \text{(Equation~\ref{eq:using_one_uniform_alg})}\\
		&= \frac{1}{3} \Ex \left[ \Ex \left[ \rank_{\w}(\M') \mid \vec  w \right] \right]\\ 	
		&\geq \frac{1}{3} \Ex \left[\alpha \cdot \rank_{\w}(\M)\right] && (\text{$\alpha$-partition property})\\
		&= \frac{\alpha}{3} \cdot \Ex\left[ \rank_{\w}(\M) \right]. 
	\end{align*}
Above, the second equality holds due to Fubini's theorem since $\rank_{\w}(\M')$ is a non-negative random variable with a finite expectation. 
We conclude that Algorithm~\,\ref{alg:pi_partition_algorithm} is $\frac{\alpha}{3}$-competitive, as needed. We note that since the competitive ratio of $\A$ holds against the almighty adversary, so does ours.
\end{proof}

\subsection{Contention Resolution}
Now, we turn our attention to proving the existence of $O(\alpha)$-balanced offline and $O(\alpha^2)$-balanced online pairwise independent CRS for matroids that satisfy $\alpha$ partition property. We begin with the offline result. 

\begin{theorem}
	\label{thm:crs-partition}
	For matroids satisfying the $\alpha$-partition property for some $\alpha \in (0,1]$, there is a $\frac {1}{1.299} \cdot \left( 1 - \frac 1 e\right) \cdot \alpha$-balanced  pairwise-independent offline CRS.
\end{theorem}
\begin{proof}
  Consider a matroid $\M=(E,\I)$ and a distribution $\D\in \Delta_{\text{pw}}(2^E)(\vec \mu)$ with $\mu \in \P(\M)$. Let $\M'=(E',\I')$ be a random partition matroid which $\alpha$-approximates $\M$ in the sense of Definition~\ref{def:alphapartition}.
  Let $\tilde{R} \sse E$ be sampled from the product distribution with marginals $\mu$; i.e.,  each element $e \in E$ is included in $\tilde{R}$ independently with probability $\mu(e)$.  It was shown in \cite{chekuri2011multi} that the class of product distributions with marginals in $\P(\M)$ admits a $(1-1/e)$-balanced offline CRS. Therefore, by Theorem~\ref{thm:shaddins characterization} the following holds for all sets of elements $F \subseteq E$.
\begin{align}\label{eq:offline CRS crucial}
	\mathbb{E}_{\tilde R}[ \rank_\mathcal M(\tilde{R} \cap F)] \geq (1-1/e) \cdot \vec \mu(F).
\end{align}
Moreover, the $\alpha$-partition property implies that for every fixed $\tilde{R}$, $\Ex_{\M'}[\rank_{\M'}(\tilde R \cap F)] \geq \alpha \cdot \rank_\M(\tilde R \cap F)$.
Combining this with Equation~\ref{eq:offline CRS crucial}, we obtain
\begin{equation}\label{eq:offline crs partition}
	\mathbb{E}_{\mathcal M',\tilde{R}} [ \rank_{\mathcal M'}( \tilde{R} \cap F) ] \geq (1-1/e) \cdot \alpha \cdot \vec \mu(F).
\end{equation}
Let $P_1,\ldots,P_k$ be the parts of the matroid $\mathcal M'$, and let $R$ be sampled from the pairwise independent distribution $\mathcal D$. We have the following for any $F \sse E$.
	\begin{align*}
		\Ex_{R} [\rank_{\M}(R \cap F)] 
		&\geq \Ex_{\M'} \left[ \Ex_{R} [\rank_{\M'}(R \cap F)] \right] && (\I' \subseteq \I)\\
		&= \Ex_{\M'} \left[ \sum_{i=1}^k \Pr_{R}\left[ \bigvee_{e \in F \cap P_i} e \in R \right] \right] && (\text{Partition matroid})\\
		&\geq \Ex_{\M'} \left[ \sum_{i=1}^k \frac{1}{1.299} \cdot \Pr_{\tilde R}\left[ \bigvee_{e \in F \cap P_i} e \in \tilde R \right] \right] && (\text{Lemma~}\ref{lem:local-lemma-type})\\
		&=\frac{1}{1.299} \cdot \mathbb E_{\M', \tilde R} [\rank_{\mathcal M'}(\tilde R\cap F)] && (\text{Partition matroid})\\
		&\geq\frac{1}{1.299} \cdot \left(1-\frac 1 e\right) \cdot \alpha \cdot \vec \mu(F). && (\text{Equation~}\ref{eq:offline crs partition})
	\end{align*}

Combining above inequality with the characterization from Theorem~\ref{thm:shaddins characterization}, we conclude that the family of pairwise independent distributions $ \Delta_{\text{pw}}(\vec \mu)$ with $\vec \mu \in \mathcal P_{\mathcal M}$ admits $\left(\frac{1}{1.29} \cdot (1-1/e) \cdot \alpha\right)$-balanced CRS when matroid $\M$ satisfies the $\alpha$-partition property.
\end{proof}

Next,  we show how to combine our offline pairwise-independent CRS with our pairwise-independent prophet inequality to obtain an online CRS. 
We use the following lemma, the proof of which closely mirrors the arguments used in the proof of Theorem 4.1 from \cite{dughmi20}. Hence, we omit its detailed presentation here.

\begin{restatable}{lemma}{reductiontooffline}\label{lem:offline_to_online_CRS_reduction}
If a matroid $\M$ admits a $\gamma$-competitive pairwise-independent prophet inequality against the almighty adversary and a $\beta$-balanced offline pairwise-independent CRS, then it also admits a $\beta \cdot \gamma$-balanced OCRS against the almighty adversary.
\end{restatable}

Combining Lemma~\ref{lem:offline_to_online_CRS_reduction} and Theorem~\ref{thm:crs-partition} and \ref{thm:PIwithPartition}, we obtain the following theorem.

\begin{theorem}
	\label{thm:ocrs-partition}
	{For matroids satisfying the $\alpha$-partition property for some $\alpha \in (0,1]$, there is a $\frac {1}{1.299} \cdot \left( 1 - \frac 1 e\right) \cdot \frac{ \alpha^2}{3}$-balanced pairwise-independent OCRS against the almighty adversary.}
\end{theorem}

\subsection{Structural Implications}

We now shift our focus to showing that full linear matroids over a finite field do not admit a partition property with strong approximation guarantees. We prove this by combining the offline CRS of Theorem~\ref{thm:crs-partition} with the impossibility result for contention resolution presented in Theorem~\ref{thm:CRSmain}.

\begin{cor}\label{cor:parition_bounds}
{The full linear matroid $\FF_q^d$ of rank $d$, with $q \geq d$, does not satisfy an $\alpha$-partition property with $\alpha = \omega\left(\frac{1}{d}\right)$. Moreover, the full binary matroid $\FF_2^d$ of rank $d$ does not satisfy an $\alpha$-partition property with $\alpha = \omega\left(\frac{\log d}{d}\right)$.}
\end{cor}
\begin{proof}
  First, we briefly argue that duplicating elements of a matroid preserves the partition property. Suppose that $\M=(E,\I)$ satisfies the $\alpha$-partition property, as witnessed by a random matroid $\M'$ as in Definition~\ref{def:alphapartition}. For a positive integer $m$, let $\M \times m = (E^{\times m}, \I^{\times m})$ be the matroid  which includes $m$ parallel duplicates $e^1,\ldots,e^m$ of each element $e$ of $\M$, as described in Section~\ref{sec:matroid-prelims}. Let $\M' \times m $  be defined similarly.   For a weight vector $\w$ indexed by the elements $E^{\times m}$ of $\M^{\times m}$, let $\tilde{\w} \in \RR^E$ be such that $\tilde{w}(e) = \max_{i=1}^m w(e^i)$. Similarly, for  $S \sse E^{\times m}$ let $\tilde{S} = \set{e \in E : e^i \in S \mbox{ for some } i}$. It is clear that $\rank_w^{\M \times m} (S) = \rank_{\tilde{\w}}^\M(\tilde{S})$, and similarly $\rank_\w^{\M' \times m} (S) = \rank_{\tilde{\w}}^{\M'}(\tilde{S})$. The following calculation for arbitrary $S \sse E^{\times m}$ shows that $\M \times m$ satisfies the $\alpha$-partition property, as witnessed by  $\M' \times m$.
			\begin{align*}
				\E_{\M'} \left[\rank_{w}^{\M' \times m} (S) \right] &= \E_{\M'} \left[\rank_{\tilde{\w}}^{\M'}(\tilde{S}) \right]  \\
				&\geq \alpha \cdot  \rank_{\tilde{\w}}^\M(\tilde{S})\\ 
				&=\alpha \cdot  \rank_{\w}^{\M\times m}(S)
			\end{align*}

	Given that the partition property is invariant to duplicating elements, combining Theorem~\ref{thm:CRSmain} and~\ref{thm:crs-partition} now yields the corollary.
\end{proof}


\section{Open Questions}
\label{sec:conclusion}

Our results indicate that pairwise independence lends insufficient structure for constant approximations to contention resolution and prophet inequalities on matroids. More generally, it is natural to investigate the same questions for $k$-wise independence, and to quantify the optimal ratios as a function of $k$. Our impossibility results rely on the recipe presented in Section~\ref{sec:tool}, which easily generalizes to $k$-wise independence only for ordered vector families. Generalizing our unordered construction to arbitrary $k$ promises to extend our results to $k$-wise independence. We conjecture an optimal bound of $O(k/\text{rank})$ for contention resolution, and refrain from  such conjecture for prophet inequalities.

Due to the equivalence between the matroid secretary problem and the matroid prophet secretary problem from \cite{dughmi22},  extending our impossibility result for prophet inequalities to the random order model would disprove the matroid secretary conjecture. In fact, constructions which are $k$-wise independent might be particularly promising, since they preclude any ``learning'' from samples of size~$k$. On the flip side, designing secretary algorithms for variants of our  construction could stimulate the development of pertinent algorithmic techniques.


\bibliographystyle{plainnat} 
\bibliography{refs}

\appendix

\section{Missing Proofs from Section~\ref{sec:tool}} \label{sec:missingproofsTool}

\lemLinearIndep*
\begin{proof}
	The proof of this lemma proceeds via induction on $m$. The base case where $m=1$ is straightforward, as $\rank(R) \geq 0$, which validates the claim. For the induction step, consider $m>1$, and let $R \in \FF_q^{d \times m}$ be a matrix generated uniformly at random with column denoted by $\vec r_1, \dots \vec r_m$. Then, we have the following:
	\begin{align*}
		&\Pr[\rank(R)=m] \\	
		&= \Pr[\rank(\vec{r}_1, \dots \vec{r}_{m-1}) = m-1] \cdot \Pr[\vec{r_m} \notin \spn(\vec{r}_1, \dots \vec{r}_{m-1}) \mid \rank(\vec{r}_1, \dots \vec{r}_{m-1}) = m-1]\\
		&\geq \left(1-\frac{1}{q^{d-m+1}}\right) \cdot \Pr[\vec{r_m} \notin \spn(\vec{r}_1, \dots \vec{r}_{m-1}) \mid \rank(\vec{r}_1, \dots \vec{r}_{m-1}) = m-1]\\
		&= \left(1-\frac{1}{q^{d-m+1}}\right) \cdot \left(1-\frac{1}{q^{d-m+1}} \right)\geq 1-\frac{2}{q^{d-m+1}}\\
		&\geq 1-\frac{q}{q^{d-m+1}}=1-\frac{1}{q^{d-m}}.
	\end{align*}
	Above, the first inequality follows from the induction on $m$. The second equality holds since $\vec r_m$ is a sampled independent of $\vec r_1,\dots , \vec r_{m-1}$ and  uniformly from $\FF_q^d \setminus \spn(\vec r_1 , \dots, \vec r_{m-1})$ once conditioned on the event $\vec r_m \notin \spn(\vec r_1,\dots , \vec r_{m-1})$. Hence, $\Pr[\vec r_m \notin \spn(\vec{r}_1, \dots \vec{r}_{m-1})\mid \rank(\vec r_1,\dots ,\vec r_{m-1})] = \left(1-\frac{1}{q^{d-m+1}} \right)$.
\end{proof}

\lemSetConstruct*

\begin{proof}
First we fix any $\vec v^i \in \FF_q^d \times [n]$. We can express,
\begin{align*}
 \Pr[\vec v^i \in A] &= \left(  1 - \frac{1}{q^d}\right) \cdot \Pr[\vec v^i \in A\mid A\sim \D_1] + \frac{1}{q^d} \cdot \Pr[\vec v^i \in A \mid A\sim \D_2]\\
 & = \left(  1 - \frac{1}{q^d}\right) \cdot \Pr[\vec x_i = \vec v] + \frac{1}{q^d} \cdot \frac{1}{q^d}\\
 & = \left(  1 - \frac{1}{q^d}\right) \cdot \frac 1{q^d}+ \frac{1}{q^d} \cdot \frac{1}{q^d} = \frac{1}{q^d}.
\end{align*}
Above the second equality follows from the definition of the distribution $\D_1$ and $\D_2$. The second equality follows because each $\vec x_i$ is sampled uniformly from $\FF_q^d$. This concludes the proof of the first part of the lemma. 

Next, we fix any two distinct elements $\vec v^i, \vec u^j \in \FF_q^d$. We compute the joint probability of $\Pr[\vec v^i \in A \land \vec u^j \in A]$ for two cases separately. First, we consider the case when $i\neq j$. We have,
\begin{align*}
	\Pr[\vec v^i \in A \land \vec u^j \in A] &= \left(  1 - \frac{1}{q^d}\right) \cdot \Pr[\vec v^i \in A \land \vec u^j \in A  \mid A\sim \D_1] + \frac{1}{q^d} \cdot \Pr[\vec v^i \in A \land \vec u^j \in A \mid A\sim \D_2]\\
	& = \left(  1 - \frac{1}{q^d}\right) \cdot \Pr[\vec x_i = \vec v\land \vec x_j = \vec u] + \frac{1}{q^d} \cdot \frac{1}{q^{2d}}\\
	& = \left(  1 - \frac{1}{q^d}\right) \cdot \frac 1{q^{2d}}+ \frac{1}{q^d} \cdot \frac{1}{q^{2d}} = \frac{1}{q^{2d}}.
\end{align*}
Above the second equality follows from the definition of the distribution $\D_1$ and $\D_2$. The second equality follows because the ordered family of vectors $\vec x_1,\dots , x_n$ are pairwise independent and $i\neq j$. 

Second, we consider the case when $i=j$. We have 
\begin{align*}
	\Pr[\vec v^i \in A \land \vec u^i \in A] &= \left(  1 - \frac{1}{q^d}\right) \cdot \Pr[\vec v^i \in A \land \vec u^i \in A  \mid A\sim \D_1] + \frac{1}{q^d} \cdot \Pr[\vec v^i \in A \land \vec u^i \in A \mid A\sim \D_2]\\
	& = \frac{1}{q^d} \cdot \frac{1}{q^{d}} =  \frac{1}{q^{2d}}.
\end{align*}
Above the second equality follows from the definition of the distribution $\D_1$ and $\D_2$. This concludes the proof of the lemma. 
\end{proof}

\end{document}
